\documentclass[12pt,draftcls,onecolumn]{IEEEtran}

\makeatletter
\def\ps@headings{%
\def\@oddhead{\mbox{}\scriptsize\rightmark \hfil \thepage}%
\def\@evenhead{\scriptsize\thepage \hfil \leftmark\mbox{}}%
\def\@oddfoot{}%
\def\@evenfoot{}}
\makeatother
\usepackage{amsfonts}
\usepackage{mathrsfs}
\usepackage{amsfonts}
\usepackage{amssymb}
\usepackage{graphicx}
\usepackage{epsfig}
\usepackage{epstopdf}
\usepackage{psfrag}
\usepackage{amsmath}
\usepackage{array}
\usepackage{subfigure}
\usepackage{cite,graphicx,amsmath,amssymb,color}
\usepackage{algorithmic}
\usepackage{algorithm}
\usepackage{algorithm}
\usepackage{algorithmic}
\usepackage{stmaryrd}
\usepackage{multirow}
\usepackage{subfig}
\usepackage{graphicx,times,amsmath} 
\usepackage{url}

\IEEEoverridecommandlockouts
\begin{document}
\bibliographystyle{IEEEtran}
\title{Structural Properties of Uncoded Placement Optimization for Coded Delivery}

\author{Sian Jin,~\IEEEmembership{Student~Member,~IEEE}\thanks{S. Jin, Y. Cui and H. Liu are with Shanghai Jiao Tong University, China. G. Caire is with Technical University of Berlin, Germany.}, \ Ying Cui,~\IEEEmembership{Member,~IEEE}, \\ Hui Liu,~\IEEEmembership{Fellow,~IEEE},  \ Giuseppe Caire,~\IEEEmembership{Fellow,~IEEE}}

\maketitle
\newtheorem{Thm}{Theorem}
\newtheorem{Lem}{Lemma}
\newtheorem{Sta}{Statement}
\newtheorem{Cor}{Corollary}
\newtheorem{Def}{Definition}
\newtheorem{Exam}{Example}
\newtheorem{Alg}{Algorithm}
\newtheorem{Sch}{Scheme}
\newtheorem{Prob}{Problem}
\newtheorem{Rem}{Remark}
\newtheorem{Proof}{Proof}
\newtheorem{Asump}{Assumption}
\newtheorem{Subp}{Subproblem}

\vspace{-1cm}

\begin{abstract}
A centralized coded caching scheme has been  proposed by Maddah-Ali and Niesen to reduce the worst-case load of a network consisting of a  server  with access to $N$ files and connected through a shared link to $K$ users, each equipped with a cache of size $M$.
However, this centralized coded caching scheme
is not able to take advantage of a non-uniform, possibly very skewed, file popularity distribution.
In this work, we consider  the same network setting but aim to reduce the average load under an arbitrary (known) file  popularity distribution.
First, we  consider  a  class of   centralized  coded caching schemes   utilizing   general  uncoded placement  and  a specific coded delivery strategy, which  are specified   by a general  file partition parameter.
Then, we  formulate the   coded caching design optimization problem over the considered class of schemes     with $2^K\cdot N^K$ variables   to minimize the average load by optimizing the file partition parameter  under an arbitrary file popularity.
Furthermore, we show that the  optimization problem is  convex, and  the resulting optimal solution  generally improves upon known schemes.
Next, we analyze structural properties of the optimization problem to obtain design insights and reduce the complexity.
Specifically, we obtain an equivalent linear optimization problem with $(K+1)N$ variables under an arbitrary  file popularity  and an equivalent linear optimization problem with $K+1$  variables under the  uniform file popularity.
Under the uniform file popularity,  we also obtain the  closed-form optimal  solution, which  corresponds to   Maddah-Ali--Niesen's centralized coded caching scheme.
Finally, we present an information-theoretic converse bound on the average load under an arbitrary file popularity.
\end{abstract}

\section{Introduction}
To support the dramatic growth of wireless data traffic, caching and multicasting  have been  proposed as two  promising approaches for massive content delivery in wireless networks.
By proactively placing content closer to or even at end-users during the off-peak hours, network congestion during the peak hours can be greatly reduced.
On the other hand,
leveraging the broadcast nature of the wireless medium by multicast transmission,
popular content can be delivered to multiple requesters simultaneously.
Recently, a new class of caching schemes for content placement in user caches, referred to as {\em coded caching}~\cite{AliFundamental, AliDec, jin, finite, YuQian, NonuniformDemands,  ji2015order, Jinbei,Sinong, KaiWan,bound17,bound16}, which jointly consider caching and multicasting,  have received significant interest.
The main novelty of such schemes with respect to  (w.r.t.) conventional approaches (e.g., as currently used in content delivery networks)
is that the messages stored in the user caches  are treated as ``receiver side information'' in order to enable network-coded multicasting, such that a
single multicast codeword is useful to a large number of users, even though they are not requesting
the same content.
In~\cite{AliFundamental} and~\cite{AliDec}, Maddah-Ali and Niesen consider a system with one server connected  through a shared error-free link to $K$ users. The server has a  library  of $N$ files  (of the same length), and each user  has an isolated cache memory of $M$ files. They formulate a caching problem, consisting of two phases, i.e., uncoded content placement and coded content delivery, which has been successively investigated in a large number of recent works~\cite{ jin, finite, YuQian, NonuniformDemands,  ji2015order, Jinbei,Sinong, KaiWan, bound17,bound16} under the same network setting.

In~\cite{AliFundamental, AliDec, jin,finite, YuQian}, the goal is to reduce the worst-case (over all possible requests) load\footnote{For future reference, in this paper, we refer to ``load'' of a particular coded caching scheme
as the ratio of the length of the  coded multicast message over the length of a single library file.} of the shared link in the delivery phase.
In particular, in~\cite{AliFundamental}, Maddah-Ali and Niesen propose a centralized coded caching scheme, which requires   knowledge of  the number of active users in the delivery phase, and achieves order-optimal memory-load tradeoff.
It is successively shown in~\cite{KaiWan} that Maddah-Ali--Niesen's centralized coded caching scheme~\cite{AliFundamental} achieves the minimum worst-case load under  uncoded placement and $N \geq K$.
Motivated by \cite{AliFundamental}, decentralized coded caching schemes are  proposed in~\cite{AliDec} and~\cite{jin},  where the number of  active users in the system  are not known in the placement phases
 and the schemes can achieve order-optimal memory-load tradeoff
in the asymptotic regime of infinite file size (i.e., the number of data units per file goes to infinity).
In the finite file size regime,  Maddah-Ali--Niesen's decentralized coded caching scheme~\cite{AliDec} is shown to achieve  an  undesirable  worst-case load~\cite{finite},
 which is larger than the  worst-case  load achieved by the decentralized scheme that we present   in~\cite{jin}.
In~\cite{YuQian},  Yu {\em et al.} propose a centralized coded caching scheme  to reduce the average load under the uniform file popularity, by efficiently serving  users with  common requests.
Note that all the coded caching schemes in~\cite{AliFundamental, AliDec, jin, YuQian} dedicate the same fraction of memory to each  file, and   may not be able to take full advantage of a non-uniform, possibly very skewed, popularity distribution.

In~\cite{NonuniformDemands, ji2015order, Jinbei, Sinong}, the goal is to reduce the average load of the shared link in the delivery phase under an arbitrary file popularity.
Specifically, in~\cite{NonuniformDemands}, the authors partition files  into multiple groups and apply Maddah-Ali--Niesen's  decentralized  coded caching scheme~\cite{AliDec} to each group.
As coded-multicasting opportunities for files from different groups are not fully explored,  it is expected that  the average load  in~\cite{NonuniformDemands} can be further reduced.
In~\cite{ji2015order}, a decentralized  coded caching  scheme where the memory allocation for files with different popularity is optimized using an upper bound on the average load is proposed. However, such optimization is highly non-convex and not amenable to analysis. Therefore, a simpler but suboptimal scheme  (referred to as the RLFU-GCC decentralized coded caching scheme)  where the library is partitioned only into two groups, is also proposed for the purpose of asymptotic analysis, and some optimality properties in the scaling laws of the  average load versus the system parameters are obtained analytically, in particular for the case of a Zipf popularity distribution. In~\cite{Jinbei},  inspired  by the RLFU-GCC decentralized coded caching scheme in~\cite{ji2015order}, Zhang {\em et al.} present a similar coded caching scheme which partitions the library  into two groups  and  show  that the achieved average load is  within a constant factor of the minimum average load over all possible schemes under an arbitrary file popularity (except a small additive term) in the general regimes of the system parameters.
In~\cite{Sinong}, Wang {\em et al.} formulate a   coded caching design problem  to minimize the average load by  optimizing  the cache memory for storing each file. To reduce the computational complexity, Wang {\em et al.} consider a simplified objective function, i.e., the total average size of the uncached files and obtain a sub-optimal solution, which is shown to be order-optimal when the number of users and the number of files are large, assuming that the file popularity follows a Zipf distribution. However, in the general regime, there is no performance guarantee for the sub-optimal solution.

Besides achievable schemes, \cite{YuQian,  NonuniformDemands, ji2015order, Jinbei, Sinong, bound16, bound17, Factor2} present information-theoretic converse bounds for coded caching.
The bounds in~\cite{YuQian,  NonuniformDemands, ji2015order, Jinbei, Sinong, bound16, bound17, Factor2} can be classified into two classes, i.e., class~i): bounds that are only suitable for  uncoded placement and class~ii): bounds that are suitable for any placement including uncoded placement and coded placement. The bound in~\cite{YuQian} belongs to class~i) and is exactly tight,  for both the worst-case load and the average load (under the uniform file popularity).
The bounds based on  reduction from an arbitrary file popularity  to the uniform file popularity~\cite{NonuniformDemands, ji2015order,Jinbei},  cut-set~\cite{Sinong}, relation between a multi-user single-request caching network  and a single-user multi-request caching network~\cite{bound16}, and other  information-theoretic approaches~\cite{bound17, Factor2}, belong to  class~ii). In particular,
the bound in~\cite{Factor2} is tighter than other bounds under the uniform file popularity, but it is rather complicated and cannot be applied directly to the case of non-uniform file popularity as the bounds in~\cite{NonuniformDemands, ji2015order,Jinbei, Sinong, bound16}.
However, the  bounds in~\cite{NonuniformDemands, ji2015order,Jinbei, Sinong, bound16} for an arbitrary popularity distribution are not generally tight in  non-asymptotic regimes of the system parameters.
Thus, it is important to obtain a tighter converse bound on the average load under  an arbitrary file popularity  in  non-asymptotic regimes of the system parameters.

In this paper, we consider the same problem setting as  in \cite{AliFundamental}.
To obtain first-order design insights, we  focus on    investigating      centralized coded caching schemes to minimize the average load under an arbitrary popularity distribution, which may later  motivate efficient  designs of    decentralized coded caching schemes.
Our main contributions are summarized below.

\begin{itemize}
\item  We  consider a  class of  centralized  coded caching schemes   utilizing   general uncoded placement  and  a specific coded delivery strategy, which  are specified   by a  general  file partition parameter.
This  class of   centralized coded caching schemes include  Maddah-Ali--Niesen's centralized coded caching scheme~\cite{AliFundamental},  each realization of Maddah-Ali--Niesen's  decentralized (random) coded caching scheme~\cite{AliDec} and   each realization of  Zhang {\em et al.}'s decentralized (random)  coded caching scheme~\cite{Jinbei}.

\item We  formulate the   coded caching design optimization problem over the considered class of schemes  with $2^K\cdot N^K$ variables  to minimize the average load by optimizing the file partition parameter  under an arbitrary file popularity.
Contrary to the optimization in~\cite{ji2015order} which is non-convex and is difficult to analyze, we show that the proposed  optimization problem  is  convex and  amenable to
analysis.
Furthermore, we show that the resulting   optimized average load  is  generally better than those of   upon known schemes and the linear  coded  delivery   procedures  of the considered class of schemes  have the same performance as the graph-coloring index coding  delivery procedure in~\cite{ji2015order}, called the  $GCC_1$  procedure  and the delivery procedure which adopts an appending method to avoid the ``bit waste" problem in~\cite{HCD}, called the  HCD procedure, when applied to the optimized file placement parameter.

\item We analyze  structural properties of the optimization problem to obtain design insights and reduce the complexity for obtaining an optimal solution.
Specifically, we obtain an equivalent linear optimization problem with $(K+1)N$ variables.
To further reduce the  complexity of the linear optimization  problem under the  uniform file popularity, we obtain an equivalent linear optimization problem with $K+1$  variables.
We also obtain the   closed-form optimal  solution under the uniform file popularity for all   $M \in \{0, \frac{N}{K}, \frac{2N}{K}, \ldots, N\}$.
This optimal  solution corresponds to   Maddah-Ali--Niesen's centralized coded caching scheme~\cite{AliFundamental}, aiming  at reducing  the worst-case load.
\item We present a  genie-aided converse bound on the average load under an arbitrary file popularity  using the genie-aided approach proposed in~\cite{ji2015order}.  When  the file popularity is uniform, the genie-aided converse bound reduces to the converse bound on the average load under the uniform file popularity derived in~\cite{bound17}.
\item  Numerical results verify the analytical results and demonstrate the promising  performance of the   optimized  parameter-based scheme.  Numerical results also show that the presented genie-aided converse bound is tighter than the converse bounds in~\cite{ji2015order,Jinbei} for any cache size  and  is tighter than the converse bounds in~\cite{bound16,Sinong}   when the cache size is modest or large.
\end{itemize}
\section{Problem Setting}\label{Sec:setting}
As in~\cite{AliFundamental,jin}, we consider a system with one server connected through a shared error-free link to $K\in\mathbb N$  users,  where $\mathbb N$ denotes the set of all natural numbers.\footnote{The problem setting is similar to the one we presented in~\cite{jin}, expect that here we consider an arbitrary file popularity and focus on minimizing the average load.}
The server has access to a  library  of $N\in\mathbb N$
files, denoted by $W_1, \ldots ,W_N$, each consisting of $F \in\mathbb N$ indivisible data units.
Let  $\mathcal{N}\triangleq \{1,2,\ldots ,N\}$ and $\mathcal{K} \triangleq \{1,2, \ldots K\}$ denote the set of file indices and the set of user indices, respectively.
Different from~\cite{AliFundamental}, we assume that each user randomly and independently requests a file in $\mathcal{N}$  according to  an arbitrary file popularity.
In particular, a user requests $W_n$ with probability $p_n \in [0,1]$, where $n \in \mathcal N$.
Thus, the file popularity distribution is given by $\mathbf p \triangleq \left(p_n\right)_{n=1}^{N}$, where $\sum_{n=1}^{N}p_n =1$.
In addition, without loss of generality, we assume $p_1 \geq p_2 \geq \ldots \geq p_N$.
Each user  has an isolated cache memory of $MF$ data units, for some real number $M \in [0,N]$.

The system operates in two phases, i.e., a placement phase and a delivery phase~\cite{AliFundamental}. In the placement phase, the users are given access to the entire  library  of $N$ files. Each user is then able to fill the content of its cache using the library.
Let $\phi_k$ denote the caching function for user $k$, which maps the files $W_1,  \ldots ,W_N$ into the cache content $Z_k\triangleq \phi_k(W_1,  \ldots ,W_N)$ for user  $k \in \mathcal{K}$.
Let $\boldsymbol \phi \triangleq \left(\phi_1, \ldots, \phi_K\right)$ denote the caching functions of all the $K$ users.
Note that $Z_k$   is of size  $MF$ data units.
Let   $\mathbf Z\triangleq \left(Z_1, \cdots, Z_K\right)$
 denote the cache contents of all the $K$ users.
In the delivery phase, each user randomly and independently requests one file in $\mathcal N$ according to file popularity distribution $\mathbf p$.
Let $D_{k}\in\mathcal N$ denote the index of the file requested by user $k \in \mathcal{K} $, and
let $\mathbf D\triangleq \left(D_1, \cdots, D_K\right)\in \mathcal N^K$ denote the requests of  all the $K$ users.
The server replies to these $K$ requests by sending a message  over the shared link, which is observed by all the $K$ users.
Let $\psi$ denote the encoding function for the server, which maps the files $W_1, \ldots ,W_N$, the cache contents $\mathbf Z$, and the requests $\mathbf D$ into the multicast message
$Y\triangleq  \psi(W_1, \ldots, W_N, \mathbf Z,\mathbf D)$
sent by the server over the shared link. Let $\mu_k$ denote the decoding function for user $k$, which maps the multicast message $Y$ received over the shared link, the cache content $Z_k$ and the request $D_k$, to the estimate
$\widehat{W}_{D_k}\triangleq \mu_k(Y,Z_k,D_k)$
of the requested file $W_{D_k}$ of user $k\in\mathcal K$.
 Let $\boldsymbol \mu \triangleq \left(\mu_1, \ldots, \mu_K\right)$ denote the decoding functions of all the $K$ users.
Each user should be able to recover its requested file from the message received over the shared link
and its cache content.
Thus, we impose the successful content delivery condition
$$\widehat{W}_{D_k} = W_{D_k}, \quad \forall \; k \in \mathcal{K}.$$
Given the cache size $M$, the cache contents $\mathbf Z$ and the requests $\mathbf D$ of all the $K$
users, let $R(K,N,M, \boldsymbol \phi,\mathbf D)F$ be the length (expressed in data units) of the multicast message $Y$,
where $R(K,N,M,\boldsymbol \phi,\mathbf D)$ represents the (normalized) load of the shared link.
Let
$$R_{\rm avg}(K,N,M,\boldsymbol \phi) \triangleq \mathbb{E}_{\mathbf D} \left[R(K,N,M,\boldsymbol \phi,\mathbf D)\right]$$
denote the average  (normalized) load of the shared link, where the average is taken over requests $\mathbf D$.\footnote{Later, we shall use slightly different notations for the average load to reflect the dependency on the specific scheme considered.}
Let
\begin{align}
R_{\rm avg}^*(K,N,M) \triangleq \min_{\boldsymbol \phi}R_{\rm avg}(K,N,M,\boldsymbol \phi) \label{eqn:min_avg_load}
\end{align}
denote the minimum average (normalized) load of the shared link.
In this paper, we adopt a  specific delivery strategy (i.e.,  the encoding function $\psi$)
and  decoding functions $\boldsymbol \mu$. Based on these, we  wish to minimize the average load of the shared link in the delivery phase under successful content delivery condition,
by  optimizing  the placement strategy (i.e., the caching functions $\boldsymbol \phi$) of uncoded placement.
As in~\cite{AliFundamental}, in this paper  we focus on studying effective centralized coded caching schemes to obtain first-order design insights. The obtained results can be extended to design efficient decentralized coded caching schemes, e.g., using the methodology we  propose in~\cite{jin}.
However, the detailed investigation of possible decentralized schemes inspired by the centralized approach of this paper is left for future work.

\section{Centralized  Coded Caching Scheme}\label{Sec:scheme}
In this section,  we first present   a  class of  centralized  coded caching schemes    utilizing   general  uncoded placement  and  a specific coded delivery strategy, which  are specified  by a  general  file partition parameter.
Then, we show that the class of  centralized coded caching schemes include the schemes  in~\cite{AliFundamental,AliDec,Jinbei}.

\subsection{Parameter-based Centralized Coded Caching}
In the uncoded placement phase, each file is partitioned into $2^K$ nonoverlapping subfiles. We label the subfiles of file $W_n$ as
$W_n=(W_{n,\mathcal{S}}: \mathcal{S} \subseteq \mathcal{K} ),$
where subfile $W_{n,\mathcal{S}}$ represents the  data units  of file $n$  stored in the cache of the users in set  $\mathcal{S}$.
We say subfile  $W_{n,\mathcal{S}}$ is of type $s$ if $|\mathcal{S}|=s$~\cite{HCD}.
Thus, the cache content at user $k$ is given by
$$Z_k=(W_{n,\mathcal{S}}:n \in \mathcal N ,k \in \mathcal{S}, \mathcal{S} \subseteq \mathcal{K} ).$$
Let $x_{n,\mathcal{S}}$ denote the ratio between the  number of  data units in  $W_{n,\mathcal{S}}$ and the  number of  data units in $W_n$ (i.e., $F$).
Let $\mathbf{x} \triangleq (x_{n,\mathcal{S}})_{n \in \mathcal{N}, \mathcal{S} \subseteq \mathcal{K}}$ denote the  file partition parameter.
Each specific choice of the file partition parameter $\mathbf{x}$ corresponds to one centralized coded caching scheme within the considered class.
This parameter is a design parameter and  will be optimized  to minimize the average  load in Section~\ref{Sec:optimization}.
Thus, $\mathbf{x}$ satisfies
\begin{align}
&0 \leq x_{n, \mathcal{S}} \leq 1 , \quad \forall \mathcal{S} \subseteq \mathcal{K}, \ n \in \mathcal{N},\label{eqn:X_range} \\
&\sum_{s=0}^{K}\sum_{\mathcal{S}\in \{\mathcal{S} \subseteq \mathcal{K}: |\mathcal{S}|=s\}}x_{n, \mathcal{S}}=1 , \quad  \forall n \in \mathcal{N}, \label{eqn:X_sum}\\
&\sum_{n=1}^{N}\sum_{s=1}^{K}\sum_{\mathcal{S} \in \{\mathcal{S} \subseteq \mathcal{K}: |\mathcal{S}|=s, k \in  \mathcal{S}\}}x_{n, \mathcal{S}} \leq M, \quad  \forall k \in \mathcal K, \label{eqn:memory_constraint}
\end{align}
where \eqref{eqn:X_sum} represents the file partition constraint and \eqref{eqn:memory_constraint} represents the cache memory constraint.
We say $\mathbf{x}$ is feasible if it satisfies \eqref{eqn:X_range}, \eqref{eqn:X_sum} and   \eqref{eqn:memory_constraint}.

In the coded delivery phase,  the $K$ users are served simultaneously using coded-multicasting.
Consider any $s\in  \{1,2,\ldots, K\}$.
We focus on a subset  of users  $\mathcal S\subseteq \mathcal{K}$ with $|\mathcal S|=s$.
Observe that every $s-1$ users in $\mathcal S$ share a subfile that is needed by the remaining user in $\mathcal S$.
More precisely, for any $k \in \mathcal S$,  the subfile $W_{D_{k},\mathcal S \setminus \{k\}}$ is requested by the user storing cache content $k$, since it is a subfile of $W_{D_{k}}$.
At the same time, it is missing at cache content $k$ since $k \notin \mathcal S \setminus \{k\}$. Finally, it is present in the cache of any user in  $\mathcal S \setminus \{k\}$.
For any subset $\mathcal S$ of cardinality $|\mathcal S|=s$, the server transmits coded multicast message
$\oplus_{k \in \mathcal S} W_{D_{k},\mathcal S \setminus \{k\}},$
where $\oplus$ denotes bitwise XOR.
All subfiles in the coded multicast message are assumed to be zero-padded to the length of the longest subfile.
For all $s\in  \{1,2,\ldots, K\}$, we  conduct the above delivery  procedure.
The multicast message $Y$  is simply the concatenation of the coded multicast messages for all $s\in  \{1,2,\ldots, K\}$.

Finally,  we formally summarize the placement  and  delivery procedures of the class of the  centralized coded caching schemes  specified  by the general file partition parameter $\mathbf{x}$   in Algorithm~\ref{alg:col}.
\begin{algorithm}[h]
\caption{Parameter-based Centralized Coded Caching}
\textbf{placement procedure}
\begin{algorithmic}[1]
\FORALL {$k \in \mathcal K$}
  \STATE $Z_k \leftarrow (W_{n,\mathcal{S}}:n \in \mathcal N ,k \in \mathcal{S}, \mathcal{S} \subseteq \mathcal{K} )$
\ENDFOR
\end{algorithmic}
\textbf{delivery procedure}
\begin{algorithmic}[1]
\FOR {$s=K,K-1,\cdots, 1$}
  \FOR {$\mathcal{S} \subseteq \mathcal{K}:|\mathcal{S}|=s$}
     \STATE server sends $\oplus_{k \in \mathcal{S}} W_{d_k,\mathcal{S}\setminus \{k\}}$
  \ENDFOR
\ENDFOR
\end{algorithmic}\label{alg:col}
\end{algorithm}

\subsection{Relations  with  Existing Schemes}\label{Sub:comparison}
We discuss  the relation between the  class of the centralized coded caching schemes  in Algorithm~\ref{alg:col} and  Maddah-Ali--Niesen's centralized~\cite{AliFundamental} and decentralized~\cite{AliDec} coded caching schemes,  the RLFU-GCC decentralized coded caching scheme~\cite{ji2015order}  as well as Zhang {\em et al.}'s decentralized coded caching scheme~\cite{Jinbei}  (the placement of which follows that of  the RLFU-GCC decentralized coded caching scheme~\cite{ji2015order}).

First, we compare the  placement procedures of the five schemes. In the placement procedure of Algorithm~\ref{alg:col}, each file is divided into  at most  $2^K$ nonoverlapping subfiles of types $0, 1, \ldots, K$, and the number of data units in  each subfile is a design parameter and can be optimized.  Note that for  a file partition parameter, if the number of data units in a subfile is zero, then there is no need to consider this subfile. Thus  $2^K$ is the maximum  number of non-overlapping subfiles  of a file. In fact, the number of non-overlapping subfiles  of a file corresponding to the optimized file partition parameter is usually much smaller than $2^K$, as shown in Section~\ref{Sec:Numerical}.
In contrast, in the placement procedure of  Maddah-Ali--Niesen's centralized coded caching scheme, each file is divided into ${K \choose \frac{KM}{N}}$ nonoverlapping subfiles  of type $\frac{KM}{N}$, and the number of data units in  each subfile is $\frac{F}{{K \choose \frac{KM}{N}}}$.
In the placement procedure of  Maddah-Ali--Niesen's decentralized coded caching scheme, each file is divided into $2^K$ nonoverlapping subfiles of types $0, 1, \ldots, K$, and the number of data units in each subfile is  random.
In the placement procedure of   the RLFU-GCC decentralized coded caching scheme and Zhang {\em et al.}'s decentralized coded caching scheme,  the whole file set is divided into two subsets. Each file in the first subset is divided into $2^K$ nonoverlapping subfiles of types $0, 1, \ldots, K$, and the number of data units in each of these subfiles is random.
On the other hand, no  file in the second subset is  divided (equivalently, each file in the second subset can be viewed as   one subfile of type $0$).

Next, we compare the delivery  procedures of the five schemes.  The delivery procedure of Algorithm~\ref{alg:col} is the same as the delivery procedure of Maddah-Ali--Niesen's decentralized coded caching scheme  and is designed for types $0, 1, \ldots, K$.
In contrast, the delivery procedure of Maddah-Ali--Niesen's centralized coded caching scheme is designed only for type $\frac{KM}{N}$,  and the delivery procedure of  Zhang {\em et al.}'s decentralized coded caching scheme is designed for types $0, 1, \ldots, \widetilde{K}$, where $\widetilde{K} \in \{0,1, \ldots, K\}$ is a random variable.
Note that the delivery procedures of the above four coded caching schemes are linear coded delivery procedures.
The delivery procedure of  the RLFU-GCC decentralized coded caching scheme adopts a  graph-coloring index coding delivery procedure designed for types $0, 1, \ldots, K$, called the $GCC_1$ procedure. The discussion of the relation with the $GCC_1$ procedure is deferred to the end of Section~\ref{Sub:formulation}.

From the above discussion, we know that  the class of the
centralized coded caching schemes  in Algorithm~\ref{alg:col} include  Maddah-Ali--Niesen's centralized coded caching scheme, each realization of Maddah-Ali--Niesen's   decentralized (random) coded caching scheme and each realization of Zhang {\em et al.}'s   decentralized (random) coded caching scheme.\footnote{Recall that~\cite{NonuniformDemands} partitions files  into multiple groups and applies Maddah-Ali--Niesen's decentralized coded caching scheme~\cite{AliDec} to each group. Thus, the class of the
centralized coded caching schemes  in Algorithm~\ref{alg:col} also include the uncoded placement and coded delivery for each group in~\cite{NonuniformDemands}.}

\section{Average Load Minimization} \label{Sec:optimization}
In this section, we first formulate the   coded caching design optimization problem over the considered class of schemes   to minimize the average load  under an arbitrary file popularity.
Then, we analyze  structural properties of the optimization problem to obtain design insights and reduce the complexity for obtaining an optimal solution.
\subsection{Problem Formulation}\label{Sub:formulation}
Consider  the class of the
centralized coded caching schemes specified by the general file partition parameter $\mathbf{x}$  in Algorithm~\ref{alg:col}.
Denote $\overline{R}(K,N,M,\mathbf{x},\mathbf D)$ as the load for serving the $K$ users with cache size $M$ under given file partition parameter $\mathbf{x}$ and  requests  $\mathbf D$.
By Algorithm~\ref{alg:col}, we have
\begin{align}
\overline{R}(K,N,M,\mathbf{x},\mathbf D)=\sum_{s=1}^{K}\sum_{\mathcal{S}\in \{\mathcal{S} \subseteq \mathcal{K}: |\mathcal{S}|=s\}} \max_{k \in \mathcal{S}}x_{D_k,\mathcal{S}\setminus\{k\}}. \label{eqn:load_1}
\end{align}
Let $\overline{R}_{\rm avg}(K,N,M,\mathbf{x})  \triangleq \mathbb{E}_{\mathbf D} \left[\overline{R}(K,N,M,\mathbf x,\mathbf D)\right]$ denote the average load for serving the $K$ users with cache size $M$ under given file partition parameter $\mathbf{x}$, where the average is taken over random requests $\mathbf D$.
Thus, we have
\begin{align}
\overline{R}_{\rm avg}(K,N,M,\mathbf{x})=\sum_{\mathbf{d} \in \mathcal{N}^K} \left(\prod_{k=1}^{K}p_{d_k}\right) \sum_{s=1}^{K}\sum_{\mathcal{S} \in \{\mathcal{S} \subseteq \mathcal{K}: |\mathcal{S}|=s\}} \max_{k \in \mathcal{S}} x_{d_k,\mathcal{S}\setminus\{k\}}, \label{eqn:average_load_1}
\end{align}
where $\mathbf{d} \triangleq (d_1, \ldots, d_K) \in \mathcal N^K$.

The file partition parameter $\mathbf{x}$ fundamentally affects the average load $\overline{R}_{\rm avg}(K,N,M,\mathbf{x})$. We would like to minimize $\overline{R}_{\rm avg}(K,N,M,\mathbf{x})$ by  optimizing $\mathbf{x}$, under the constraints on $\mathbf{x}$ in \eqref{eqn:X_range}, \eqref{eqn:X_sum} and \eqref{eqn:memory_constraint}.
\begin{Prob}[File Partition Parameter Optimization]\label{Prob:original}
\begin{align}
\overline{R}^*_{\rm avg}(K,N,M) \triangleq \min_{\mathbf{x}} \quad &  \overline{R}_{\rm avg}(K,N,M,\mathbf{x})\nonumber \\
s.t. \quad  &\eqref{eqn:X_range}, \eqref{eqn:X_sum}, \eqref{eqn:memory_constraint},\nonumber
\end{align}
where $\overline{R}_{\rm avg}(K,N,M,\mathbf{x})$ is given by \eqref{eqn:average_load_1} and the optimal solution is denoted as $\mathbf{x}^* \triangleq (x^*_{n,\mathcal{S}})_{n \in \mathcal{N}, \mathcal{S} \subseteq \mathcal{K}}$.
\end{Prob}

The objective function of Problem~\ref{Prob:original} is convex, as it is a positive weighted sum of convex piecewise linear functions~\cite{boyd2004convex}.
In addition, the constraints of Problem~\ref{Prob:original} are linear.
Hence, Problem~\ref{Prob:original} is a convex optimization problem and  can be solved using standard convex optimization techniques.
Note that the number of variables in Problem~\ref{Prob:original} is $2^K \cdot N^K$.
Thus,  the complexity   of  Problem~\ref{Prob:original}  is huge, especially when $K$ and $N$ are large.
In Section~\ref{Sub:arbitraty} and  Section~\ref{Sub:uniform}, we shall focus on deriving equivalent simplified formulations for Problem~\ref{Prob:original} to facilitate low-complexity optimal solutions under an arbitrary  popularity distribution and the uniform popularity distribution, respectively.

Next, we  discuss the  relation between the class of the centralized coded caching schemes specified by the general file partition parameter $\mathbf{x}$ in Algorithm~\ref{alg:col}  with the coded caching  schemes in~\cite{AliFundamental,AliDec,NonuniformDemands,ji2015order,Jinbei}.
Based on the discussion in  Section~\ref{Sub:comparison},  we can make the following statement.
\begin{Sta}[Relations  with Schemes in \cite{AliFundamental,AliDec, NonuniformDemands, Jinbei}]\label{Sta:perforamnce_comparison}
The   optimized  average load  for  the class of the centralized coded caching schemes  is no greater than those of the schemes in~\cite{AliFundamental,AliDec, NonuniformDemands, Jinbei}.
\end{Sta}

Finally, we discuss  the relation between the delivery procedure in  Algorithm~\ref{alg:col} and the  $GCC_1$  procedure   in~\cite{ji2015order}.
\begin{Lem}[Relation  with  $GCC_1$ in~\cite{ji2015order}]\label{Lem:ji}
For all file partition parameters, under the placement procedure in Algorithm~\ref{alg:col}, the delivery procedure in  Algorithm~\ref{alg:col}  achieves the same average load as the  $GCC_1$ procedure.
\end{Lem}

\begin{proof} Please refer to Appendix A.
\end{proof}

From Lemma~\ref{Lem:ji}, we  know that the  $GCC_1$  procedure achieves the same average  load as the  delivery procedure in Algorithm~\ref{alg:col} at the  optimized  file partition parameter $\mathbf{x}^* $.  The  discussion on the relation   with the HCD procedure is deferred to the end of  Section~\ref{Sub:arbitraty}.
\subsection{Optimization under Arbitrary File Popularity}\label{Sub:arbitraty}
In this part, we first characterize  structural properties of Problem~\ref{Prob:original} and  simplify it  without losing optimality via two steps.
In step 1, we  show  an important structural property of  Problem~\ref{Prob:original}.
\begin{Thm}[Symmetry w.r.t. Type]\label{Thm:symmetry}
For all $n \in \mathcal{N}$ and $s \in \{0,1, \cdots, K\}$, the  values of $x^*_{n, \mathcal{S}}, \mathcal{S} \in \{\mathcal{S} \subseteq \mathcal{K}: |\mathcal{S}|=s\}$ are the  same.
\end{Thm}
\begin{proof}Please refer to Appendix B.
\end{proof}

Theorem~\ref{Thm:symmetry} indicates that at the optimal solution to Problem~\ref{Prob:original},
for all $n \in \mathcal{N}$ and $s \in \{0,1, \cdots, K\}$,
subfiles $W_{n,\mathcal{S}}$, $\mathcal{S}\in \{\mathcal{S} \subseteq \mathcal{K}: |\mathcal{S}|=s\}$ have the same size.
By Theorem~\ref{Thm:symmetry}, without losing optimality, we can set
\begin{align}
x_{n,\mathcal{S}}=y_{n,s}, \quad  \forall \mathcal{S} \subseteq \mathcal{K}, \  n \in \mathcal{N}, \label{eqn:symmetry}
\end{align}
when solving Problem~\ref{Prob:original}, where $s=|\mathcal{S}| \in \{0,1,\cdots, K\}$.
Here,   $y_{n,s}$ can be viewed as the ratio between the  number of  data units in each subfile of file $W_n$ which is of   type $s$ and the  number of  data units in file $W_n$ (i.e., $F$).
Let $\mathbf{y} \triangleq (y_{n,s})_{n \in \mathcal{N}, s \in \{0,1,\cdots, K\}}$.

By \eqref{eqn:symmetry},  the constraints in  \eqref{eqn:X_range}, \eqref{eqn:X_sum} and \eqref{eqn:memory_constraint} of Problem~\ref{Prob:original} can be converted into the following constraints:
\begin{align}
&0 \leq y_{n, s} \leq 1 , \quad  \forall  s \in \{0,1,\cdots, K\}, \ n \in \mathcal{N}, \label{eqn:X_range_2} \\
&\sum_{s=0}^{K} {K \choose s} y_{n, s}=1 , \quad  \forall n \in \mathcal{N}, \label{eqn:X_sum_2}\\
&\sum_{n=1}^{N}\sum_{s=1}^{K} {K-1 \choose s-1} y_{n, s} \leq M. \label{eqn:memory_constraint_2}
\end{align}
On the other hand, by  \eqref{eqn:symmetry}, the objective function of Problem~\ref{Prob:original} in \eqref{eqn:average_load_1} can be rewritten as
\begin{align}
\overline{R}_{\rm avg}(K,N,M,\mathbf{x})  = \sum_{\mathbf{d} \in \mathcal{N}^K} \left(\prod_{k=1}^{K}p_{d_k}\right) \sum_{s=1}^{K}\sum_{\mathcal{S}\in \{\mathcal{S} \subseteq \mathcal{K}: |\mathcal{S}|=s\}} \max_{k \in \mathcal{S}}y_{d_k,s-1}  \triangleq   \widetilde{R}_{\rm avg}(K,N,M,\mathbf{y}). \label{eqn:average_load_2}
\end{align}
Define $\mathcal{D}_{n,s} \triangleq \{n,n+1,\ldots,N\}^s\setminus \{n+1,n+2,\ldots,N\}^s$,  which represents  the set of all $s$-tuples with elements in $\left\{n,n+1,...,N\right\}$ that contain at least once the element $n$.  We  further simplify $\widetilde{R}_{\rm avg}(K,N,M,\mathbf{y})$ in~\eqref{eqn:average_load_2}.
\begin{Lem}[Simplification] \label{Lem:s symmetric}
$\widetilde{R}_{\rm avg}(K,N,M,\mathbf{y})$ in \eqref{eqn:average_load_2} is equivalent to
\begin{align}
\widetilde{R}_{\rm avg}(K,N,M,\mathbf{y})=&\sum_{s=1}^{K}{K \choose s}\sum_{n=1}^{N}\sum_{(d_1, \ldots, d_s) \in \mathcal{D}_{n,s}} \left(\prod_{k=1}^{s}p_{d_k}\right)    \max_{k \in \{1,2,\cdots, s\}}y_{d_k,s-1}.\label{eqn:average_load_3}
\end{align}
\end{Lem}
\begin{proof} Please refer to Appendix C.
\end{proof}

Note that the above transformations will not cause any optimality loss for Problem~\ref{Prob:original}. Thus, Problem~\ref{Prob:original} is equivalent to the following optimization problem.
\begin{Prob}[Equivalent Optimization in Step 1]\label{Prob:equivalent_1}
\begin{align}
\overline{R}^*_{\rm avg}(K,N,M)= \min_{\mathbf{y}} \quad &  \widetilde{R}_{\rm avg}(K,N,M,\mathbf{y})\nonumber \\
s.t. \quad  &\eqref{eqn:X_range_2}, \eqref{eqn:X_sum_2}, \eqref{eqn:memory_constraint_2}, \label{eqn:equivalent_1}
\end{align}
where $\widetilde{R}_{\rm avg}(K,N,M,\mathbf{y})$ is given by \eqref{eqn:average_load_3} and the optimal solution is denoted as $\mathbf{y}^* \triangleq (y^*_{n,s})_{n \in \mathcal{N}, s \in \{0,1,\cdots, K\}}$.
\end{Prob}

The objective function of Problem~\ref{Prob:equivalent_1} is  convex, as it is a positive weighted sum of convex piecewise linear functions~\cite{boyd2004convex}.
In addition, the constraints of Problem~\ref{Prob:equivalent_1} are linear.
Hence, Problem~\ref{Prob:equivalent_1} is a convex optimization problem and  can be solved using standard convex optimization techniques.
Note that the number of variables in Problem~\ref{Prob:equivalent_1} is  $\frac{N(N^K-1)}{N-1}$, which is much smaller than that of Problem~\ref{Prob:original} (i.e., $2^K \cdot N^K$). However, the complexity   of  Problem~\ref{Prob:equivalent_1} is still huge, especially when $K$ and $N$ are large.

In step 2,  we characterize an important structural property of Problem~\ref{Prob:equivalent_1}.



\begin{Thm}[Monotonicity w.r.t. File Popularity]\label{Thm:popularity}
For all  $n \in \{1,2,\ldots, N-1\}$ and $s \in \{1,2, \cdots, K\}$,  when $p_n \geq p_{n+1}$,
\begin{align}
y^*_{n,s} \geq  y^*_{n+1,s}. \label{eqn:opt_Y_o}
\end{align}
\end{Thm}
\begin{proof} Please refer to Appendix D.
\end{proof}

Theorem~\ref{Thm:popularity} indicates that, at the optimal solution
to Problem~\ref{Prob:equivalent_1},
for all  $n \in \{1,2,\ldots, N-1\}$ and $s \in \{1,2, \cdots, K\}$, when $p_n \geq p_{n+1}$,  the size
 of subfiles $W_{n,\mathcal{S}}$, $\mathcal{S}\subseteq \{\mathcal{S} \subseteq \mathcal{K}: |\mathcal{S}|=s\}$ is no smaller than that of subfiles $W_{n+1,\mathcal{S}}$, $\mathcal{S}\subseteq \{\mathcal{S} \subseteq \mathcal{K}: |\mathcal{S}|=s\}$.

Based on Theorem~\ref{Thm:popularity}, we have the following result.
\begin{Cor}\label{Cor:group}
For all $n_1, n_2 \in \mathcal N$ and $s \in \{0,1,\ldots, K\}$,   $y^*_{n_1,s}=y^*_{n_2,s}$ if and only if $$\sum_{s=1}^{K} {K-1 \choose s-1} y^*_{n_1, s}=\sum_{s=1}^{K} {K-1 \choose s-1} y^*_{n_2, s}.$$
\end{Cor}
\begin{proof} Please refer to Appendix E.
\end{proof}

Corollary~\ref{Cor:group} indicates that, at the optimal solution to Problem~\ref{Prob:equivalent_1}, for all $n_1, n_2 \in \mathcal N$, $s \in \{0,1,\ldots, K\}$ and   $\mathcal{S}\subseteq \{\mathcal{S} \subseteq \mathcal{K}: |\mathcal{S}|=s\}$, subfiles $W_{n_1,\mathcal{S}}$ and $W_{n_2,\mathcal{S}}$ have the same size if and only if files $n_1$ and  $n_2$ are allocated the same amount of cache memory.

In addition, by Theorem~\ref{Thm:popularity}, without losing optimality, we can include the following constraint
\begin{align}
y_{n,s} \geq y_{n+1,s},  \quad  \forall  s \in \{1, 2, \cdots, K\}, \ n \in\{1,2,\cdots,N-1\}, \label{eqn:additional_cons}
\end{align}
when solving Problem~\ref{Prob:equivalent_1}.
We now further simplify \eqref{eqn:average_load_3} based on \eqref{eqn:additional_cons}.
\begin{Lem}[Simplification]\label{Lem:simplification pop}
$\widetilde{R}_{\rm avg}(K,N,M,\mathbf{y})$ in \eqref{eqn:average_load_3}    is equivalent to
\begin{align}
\widetilde{R}_{\rm avg}(K,N,M,\mathbf{y})=\sum_{s=1}^{K}{K \choose s}\sum_{n=1}^{N}\left(\left(\sum_{n'=n}^{N}p_{n'}\right)^s-\left(\sum_{n'=n+1}^{N}p_{n'}\right)^s\right) y_{n,s-1}.\label{eqn:average_load_4}
\end{align}
\end{Lem}
\begin{proof} Please refer to Appendix F.
\end{proof}

Based on the above analysis,   Problem~\ref{Prob:equivalent_1} is equivalent to the following optimization problem.
\begin{Prob}[Equivalent Optimization in Step 2]\label{Prob:simplify_2}
\begin{align}
\overline{R}^*_{\rm avg}(K,N,M)= \min_{\mathbf{y}} \quad &  \widetilde{R}_{\rm avg}(K,N,M,\mathbf{y})\nonumber \\
s.t. \quad  &\eqref{eqn:X_range_2}, \eqref{eqn:X_sum_2}, \eqref{eqn:memory_constraint_2}, \eqref{eqn:additional_cons}, \nonumber
\end{align}
where $\widetilde{R}_{\rm avg}(K,N,M,\mathbf{y})$ is given by \eqref{eqn:average_load_4}.
\end{Prob}

The objective function  of Problem~\ref{Prob:simplify_2} is  linear.
In addition, the constraints of Problem~\ref{Prob:simplify_2} are linear.
Hence, Problem~\ref{Prob:simplify_2} is a linear optimization problem and  can be solved using linear optimization techniques.
Note that the number of variables in Problem~\ref{Prob:simplify_2} is  $(K+1)N$.   The complexity  for solving   Problem~\ref{Prob:simplify_2} using the algorithm in~\cite{LeeS13a} is $\mathcal O \left(\sqrt{(K+1)N}\right)$.

Next, we  discuss the relation between  the  delivery procedure in Algorithm~\ref{alg:col}  and the  HCD  procedure under any file partition parameter satisfying  \eqref{eqn:symmetry} and \eqref{eqn:additional_cons}.
Recall that  in the delivery procedure in  Algorithm~\ref{alg:col},  all subfiles in one coded multicast message are  zero-padded to the length of the longest subfile in the coded multicast message, which may cause the  ``bit waste" problem and degrade the average load performance.
In~\cite{HCD}, an efficient delivery procedure for coded caching, called   HCD  procedure, is  proposed, without specifying a placement procedure.
The  HCD procedure adopts an appending method  to address the ``bit waste" problem. In particular, all subfiles in one  coded multicast message are padded with  bits  from some  subfiles with larger $s$  to achieve the same length as the longest subfile in the coded multicast message.
Accordingly, these appended bits are then removed from the subfiles  with larger $s$  and do not need to be considered again when later coding these subfiles.
Thus, one may expect that the HCD procedure can  achieve a lower average load than the  delivery procedure in  Algorithm~\ref{alg:col} under any file partition parameter.
However, the following lemma shows a different result.
\begin{Lem}[Relation   with  HCD in  \cite{HCD}]\label{Lem:HCD}
For all file partition parameters satisfying  \eqref{eqn:symmetry} and \eqref{eqn:additional_cons},  under the placement procedure in Algorithm~\ref{alg:col}, the   delivery procedure in Algorithm~\ref{alg:col}  achieves the same average  load as the  HCD procedure.
\end{Lem}
\begin{proof}
Please refer to Appendix G.
\end{proof}

Recall that Theorem~\ref{Thm:symmetry} and Theorem~\ref{Thm:popularity} imply that the optimized  file partition parameter satisfies \eqref{eqn:symmetry} and \eqref{eqn:additional_cons}. Thus, Lemma~\ref{Lem:HCD} also indicates that the HCD procedure achieves the same average  load as the  delivery procedure in  Algorithm~\ref{alg:col} at the optimized  file partition parameter.
However, the HCD procedure has higher complexity than the  delivery procedure in  Algorithm~\ref{alg:col} due to the involved appending method.  By Lemma~\ref{Lem:ji}, we can also know that at the optimized file partition parameter, the HCD procedure  achieves the same average  load as the  $GCC_1$ procedure.

\subsection{Optimization  under Uniform File Popularity}\label{Sub:uniform}
In this part, we consider the uniform file popularity, i.e.,  $p_1 = p_2 = \ldots = p_N$,  and  would like to characterize another structural property and  further simplify Problem~\ref{Prob:simplify_2} in this case.
First, we show  an important structural property of Problem~\ref{Prob:simplify_2}.
\begin{Thm}[Symmetry w.r.t. File]\label{Thm:popularity_2}
For all  $n \in \{1,2,\ldots, N-1\}$ and $s \in \{1,2, \cdots, K\}$,  when $p_n=p_{n+1}$,
\begin{align}
y^*_{n,s} = y^*_{n+1,s}. \label{eqn:opt_Y_o_eq}
\end{align}
\end{Thm}
\begin{proof} Please refer to Appendix D.
\end{proof}

Theorem~\ref{Thm:popularity_2} indicates that,  at the optimal solution
to Problem~\ref{Prob:simplify_2},   for all   $n \in \{1,2,\ldots, N-1\}$ and   $s \in \{1,2, \cdots, K\}$,  when $p_n=p_{n+1}$,
the size of subfiles $W_{n,\mathcal{S}}$, $\mathcal{S}\subseteq \{\mathcal{S} \subseteq \mathcal{K}: |\mathcal{S}|=s\}$ is the same as that of subfiles $W_{n+1,\mathcal{S}}$, $\mathcal{S}\subseteq \{\mathcal{S} \subseteq \mathcal{K}: |\mathcal{S}|=s\}$.
By Theorem~\ref{Thm:popularity_2}, without losing optimality, we can set
\begin{align}
y_{n,s}=z_s, \quad  \forall  s \in \{0,1,\cdots, K\}, \ n \in \mathcal{N}, \label{eqn:symmetry_2}
\end{align}
when solving Problem~\ref{Prob:simplify_2}.
Here,   $z_s$  can be viewed as the ratio between the  number of  data units in each subfile of  type $s$ of any file  and the  number of  data units in any file (i.e., $F$).
Let $\mathbf{z} \triangleq (z_{s})_{ s \in \{0,1,\cdots, K\}}$.
By \eqref{eqn:symmetry_2}, the  constraints in  \eqref{eqn:X_range_2}, \eqref{eqn:X_sum_2} and  \eqref{eqn:memory_constraint_2} of Problem~\ref{Prob:simplify_2} can be converted into the following constraints:
\begin{align}
&0 \leq z_{s} \leq 1 , \quad s \in \{0,1,\cdots, K\}, \label{eqn:X_range_3}\\
&\sum_{s=0}^{K} {K \choose s} z_{s}=1 , \label{eqn:X_sum_3}\\
&\sum_{s=0}^{K} {K \choose s}s z_{s} \leq \frac{KM}{N}. \label{eqn:memory_constraint_3}
\end{align}
On the other hand, by \eqref{eqn:symmetry_2}, the objective function of Problem~\ref{Prob:simplify_2} in  \eqref{eqn:average_load_4} can be rewritten as
\begin{align}
&\widetilde{R}_{\rm avg}(K,N,M,\mathbf{y})= \sum_{s=0}^{K}{K \choose s}\frac{K-s}{s+1} z_{s} \triangleq   \widehat{R}_{\rm avg}(M,K,N,\mathbf{z}). \label{eqn:uniform_obj}
\end{align}

Based on the above analysis, under the uniform file popularity, Problem~\ref{Prob:simplify_2} is equivalent to the following  problem.
\begin{Prob}[Optimization under Uniform File Popularity]\label{Prob:equivalent_3}
\begin{align}
\widehat{R}^*_{\rm avg}(K,N,M)\triangleq \min_{\mathbf{z}} \quad &  \widehat{R}_{\rm avg}(K,N,M,\mathbf{z})\nonumber \\
s.t. \quad  &\eqref{eqn:X_range_3}, \eqref{eqn:X_sum_3}, \eqref{eqn:memory_constraint_3}, \nonumber
\end{align}
where $\widehat{R}_{\rm avg}(K,N,M,\mathbf{z})$  is given by \eqref{eqn:uniform_obj} and
the optimal solution is denoted as $\mathbf{z}^* \triangleq (z^*_{s})_{ s \in \{0,1,\cdots, K\}}$.
\end{Prob}

The objective function  of Problem~\ref{Prob:equivalent_3} is  linear.
In addition, the constraints of Problem~\ref{Prob:equivalent_3} are linear.
Hence, Problem~\ref{Prob:equivalent_3} is a linear optimization problem and  can be solved using linear optimization techniques.
Note that the number of variables in Problem~\ref{Prob:equivalent_3} is  $K+1$.  The complexity  for solving  Problem~\ref{Prob:equivalent_3}  using the algorithm in~\cite{LeeS13a} is $\mathcal O \left(\sqrt{K+1}\right)$.

Next, we
 discuss the relation between the centralized coded caching schemes   in Algorithm~\ref{alg:col} with  Maddah-Ali--Niesen's centralized coded caching scheme~\cite{AliFundamental}.
Maddah-Ali--Niesen's centralized coded caching  scheme focuses on  cache size $M \in \{0, \frac{N}{K}, \frac{2N}{K}, \ldots, N\}$, so that $\frac{KM}{N}$ is an integer in  $\{0, 1, \ldots, K\}$. For general $M \in [0,N]$, the worst-case load can be achieved by memory sharing.
For purpose of comparison, we consider the cache size $M \in \{0, \frac{N}{K}, \frac{2N}{K}, \ldots, N\}$.
Using KKT conditions, we have the following result.
\begin{Lem}[Optimal Solution to Problem~\ref{Prob:equivalent_3}]\label{Lem:Ali}
For cache size $M \in \left\{0, \frac{N}{K}, \frac{2N}{K}, \ldots, N\right\}$, the unique optimal solution $\mathbf{z}^*$ to  Problem~\ref{Prob:equivalent_3} is given by
\begin{align}
z^*_{s}=
\begin{cases}
\frac{1}{{K \choose \frac{KM}{N}}}, & s=\frac{KM}{N}\\
0, &s \in \{0,1,\cdots, K\}\setminus \{\frac{KM}{N}\},
\end{cases}\label{eqn:uniform_caching}
\end{align}
and the  optimal value of Problem~\ref{Prob:equivalent_3} is given by
\begin{align}
\widehat{R}^*_{\rm avg}(K,N, M)=\frac{K(1-M/N)}{1+KM/N}. \label{eqn:Ali}
\end{align}
\end{Lem}
\begin{proof} Please refer to Appendix H.
\end{proof}

Lemma~\ref{Lem:Ali} indicates that, under the uniform file popularity,  for any  cache size $M \in \left\{0, \frac{N}{K}, \frac{2N}{K}, \ldots, N\right\}$,  the  optimized  file partition parameter $\mathbf{z}^*$ in  \eqref{eqn:uniform_caching} and the  optimized  average load  $\widehat{R}^*_{\rm avg}(K,N,M)$ in~\eqref{eqn:Ali}   for  the class of the centralized coded caching schemes  are the  same as the file partition parameter  and the worst-case load of   Maddah-Ali--Niesen's centralized  coded caching scheme~\cite{AliFundamental}.
Note that this optimality only holds for the worst-case when $N \geq  K$~\cite{KaiWan}.
For the  worst-case or the case under the uniform file popularity, the optimal load  is given in~\cite{YuQian} where a scheme which improves over the original Maddah-Ali--Niesen's centralized  coded caching scheme for uncoded placement is used.

\section{Converse Bound}\label{Sec:Converse}
In this section, we present an information-theoretic  converse bound on the average load under an arbitrary file popularity.
Denote  $R^{\text{lb}}_{\rm unif}(K,N,M)$  as the converse bound on the average load under the uniform file popularity obtained in~\cite{bound17}, where
\begin{align}
&R^{\text{lb}}_{\rm unif}(K,N,M) \triangleq \nonumber \\
& \max \left\{\max_{l \in \{1, \ldots, K\}}\left(1-\left(1-\frac{1}{N}\right)^l\right)\left(N-lM\right),
\max_{l \in \{1, \ldots, K\}}\left(\left(1-\left(1-\frac{1}{N}\right)^l\right)N-\frac{l\left(l+1\right)}{2N}M\right)\right\}.\label{eqn:R_uniform_lb}
\end{align}
Using  the genie-aided approach proposed in~\cite{ji2015order} and the converse bound $R^{\text{lb}}_{\rm unif}(K,N,M)$ derived  in~\cite{bound17}, we have the following result.
\begin{Lem}[Genie-aided Converse Bound]\label{Lem:converse}
For all $N\in\mathbb N$, $K\in\mathbb N$ and  $M \in [0,N]$, the minimum average load in~\eqref{eqn:min_avg_load} satisfies
\begin{align}
R_{\rm avg}^*(K,N,M) &\geq R_{\rm avg}^{\text{lb}}(K,N,M)\nonumber  \\
&\triangleq \max_{N' \in \{1,2,\ldots, N\}} \sum_{K'=1}^{K} {K \choose K'}\left(N'p_{N'}\right)^{K'}\left(1-N'p_{N'}\right)^{K-K'}R^{\text{lb}}_{\rm unif}(K',N',M), \label{eqn:R_avg_lb}
\end{align}
where  $R^{\text{lb}}_{\rm unif}(\cdot)$  is given by~\eqref{eqn:R_uniform_lb}.
\end{Lem}

The difference between Lemma~\ref{Lem:converse} and Theorem 2 in~\cite{ji2015order} lies in   the fact that the two results utilize two different converse bounds  on the  average load under the uniform file popularity.
In particular,  Lemma~\ref{Lem:converse} adopts the converse bound on the   average load under the uniform file popularity derived in~\cite{bound17}, while Theorem 2 in~\cite{ji2015order} utilizes a converse bound on the  average load under the uniform file popularity, which is derived  using a self-bounding function~\cite{ji2015order}. The purpose of replacing the converse bound for the uniform file popularity in~\cite{bound17} with the one in~\cite{ji2015order} is to obtain a tighter converse bound for an arbitrary file popularity. Later, in Section~\ref{Sec:Numerical}, we shall see that the presented converse bound is indeed tighter than that in~\cite{ji2015order}, using numerical results.

From Lemma~\ref{Lem:converse}, we know that under the uniform file popularity,  $R_{\rm avg}^{\text{lb}}(K,N,M) = R^{\text{lb}}_{\rm unif}(K,N,M)$. This means that when the  file popularity is uniform, the genie-aided converse bound in~\eqref{eqn:R_avg_lb} reduces to the converse bound on the average load under the uniform file popularity derived in~\cite{bound17}.

\section{Numerical Results}\label{Sec:Numerical}
In this section, using numerical results, we first demonstrate special properties of the proposed optimal solutions. Then, we compare the proposed optimal solution with existing solutions. Finally, we compare the presented genie-aided  converse bound with  existing information-theoretic converse bounds.
In the simulation, as in~\cite{ji2015order,HCD,NonuniformDemands,Sinong}, we assume the file popularity follows Zipf distribution, i.e., $p_n=\frac{n^{-\gamma}}{\sum_{n \in \mathcal N} n^{-\gamma}}$ for all $n\in \mathcal N$, where $\gamma$ is the Zipf exponent.
\begin{figure}
\begin{center}
    \subfigure[\small{$s=2, \gamma=1$.}]
  {\resizebox{4cm}{!}{\includegraphics{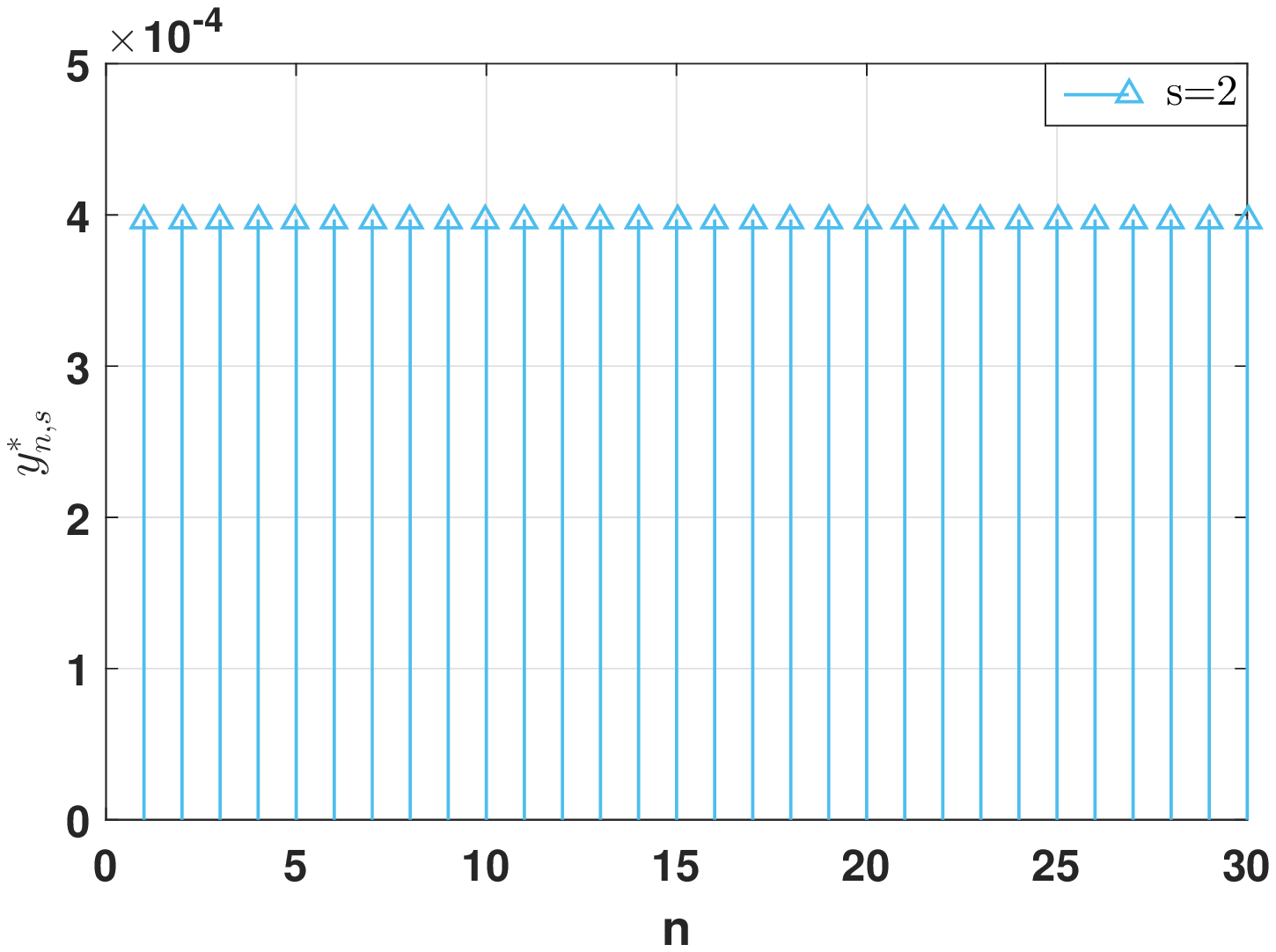}}}
      \subfigure[\small{$s=3, \gamma=1$.}]
  {\resizebox{4cm}{!}{\includegraphics{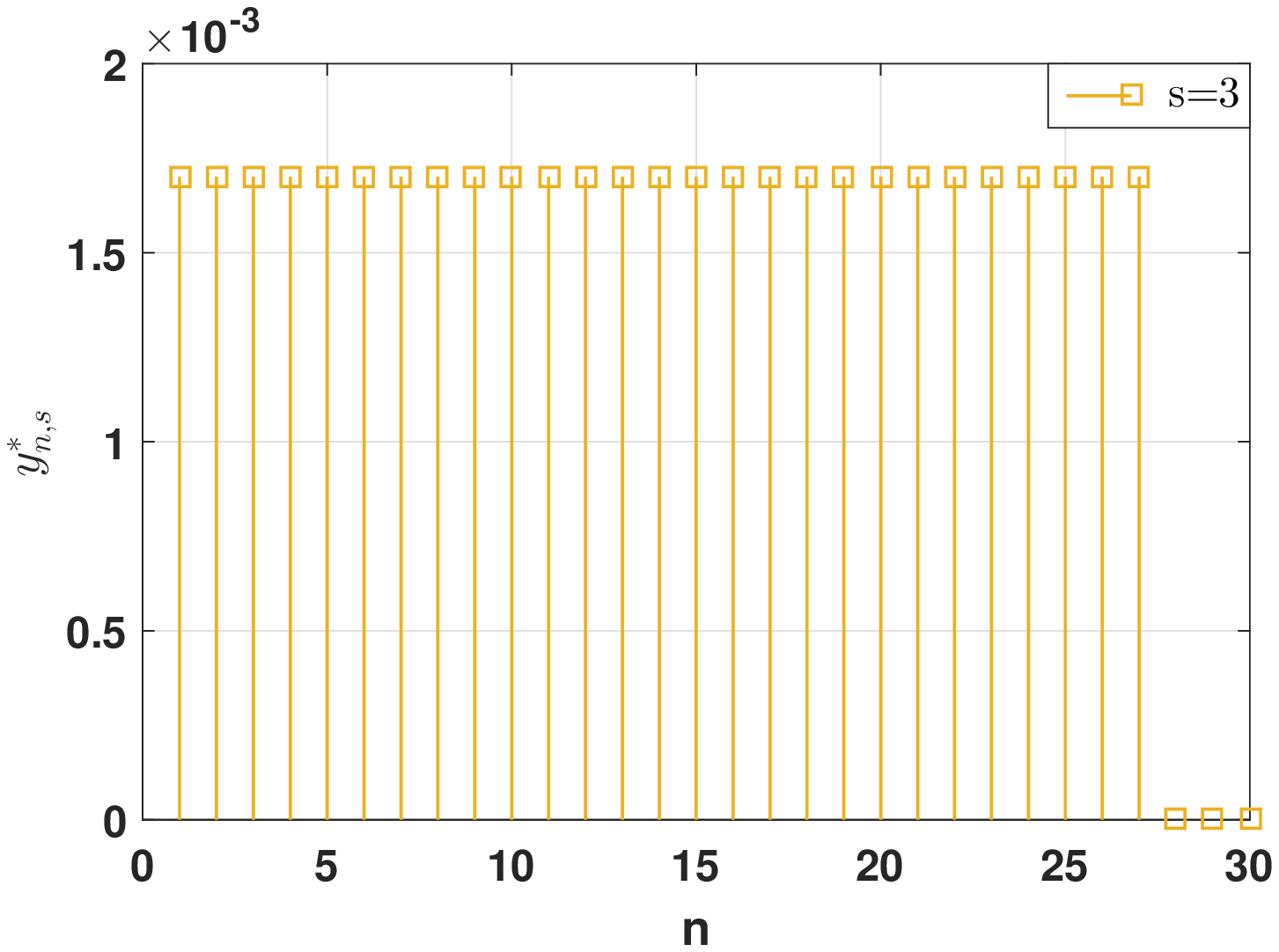}}}
        \subfigure[\small{$s=0, \gamma=1$.}]
      {\resizebox{4cm}{!}{\includegraphics{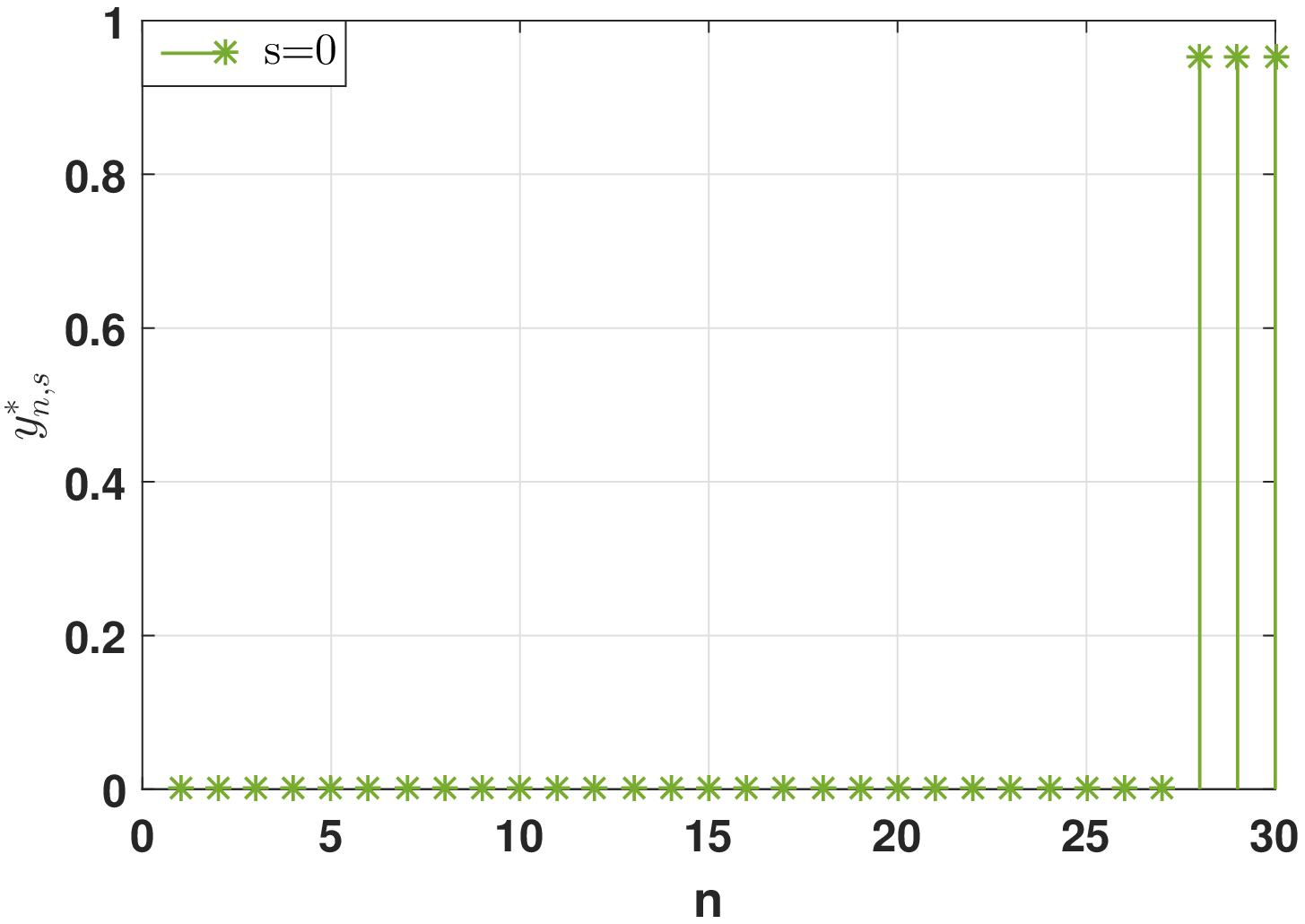}}}
          \subfigure[\small{$\gamma=1$.}]
    {\resizebox{4cm}{!}{\includegraphics{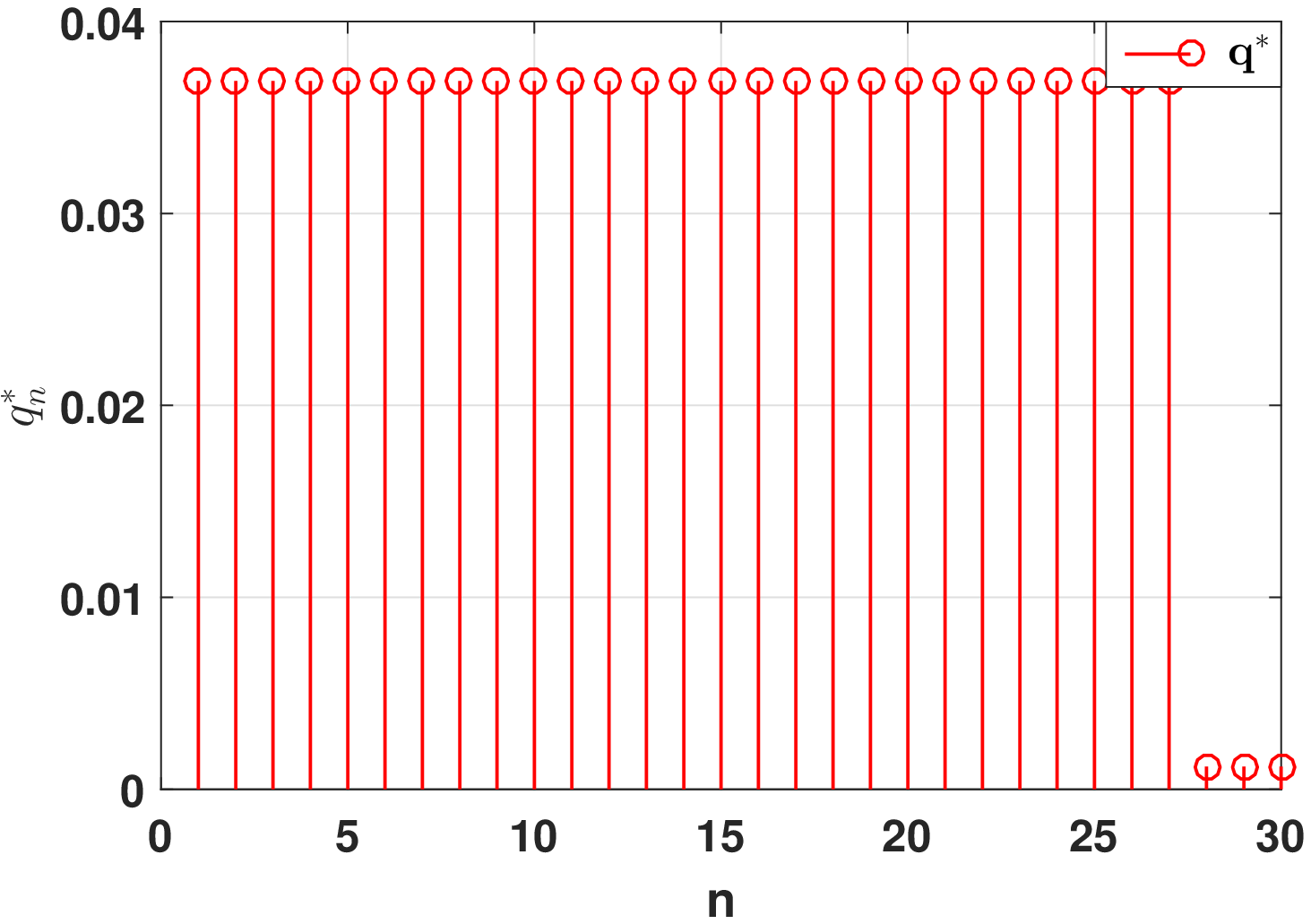}}}
      \subfigure[\small{$s=7, \gamma=1.5$.}]
        {\resizebox{4cm}{!}{\includegraphics{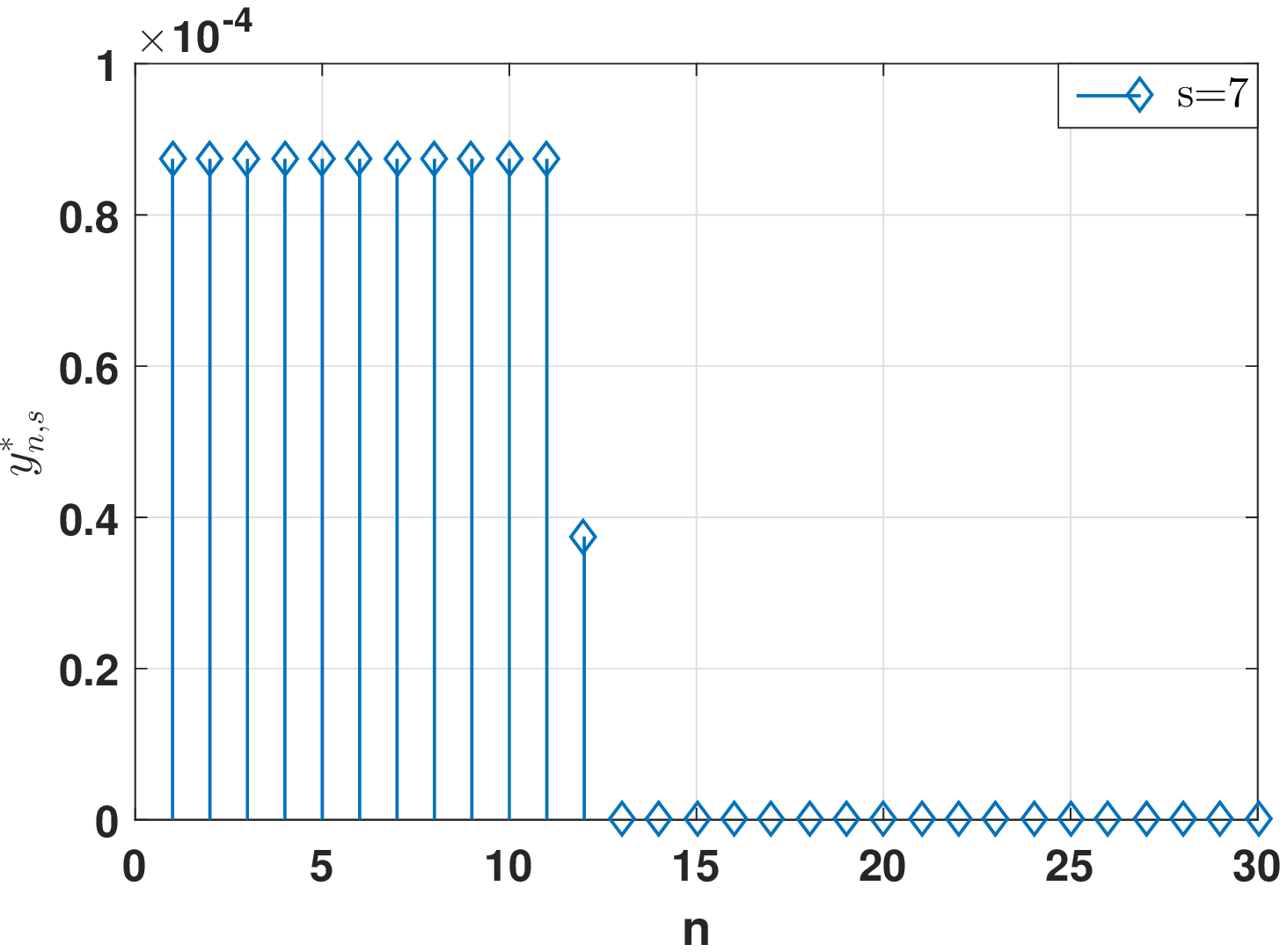}}}
        \subfigure[\small{$s=0, \gamma=1.5$.}]
      {\resizebox{4cm}{!}{\includegraphics{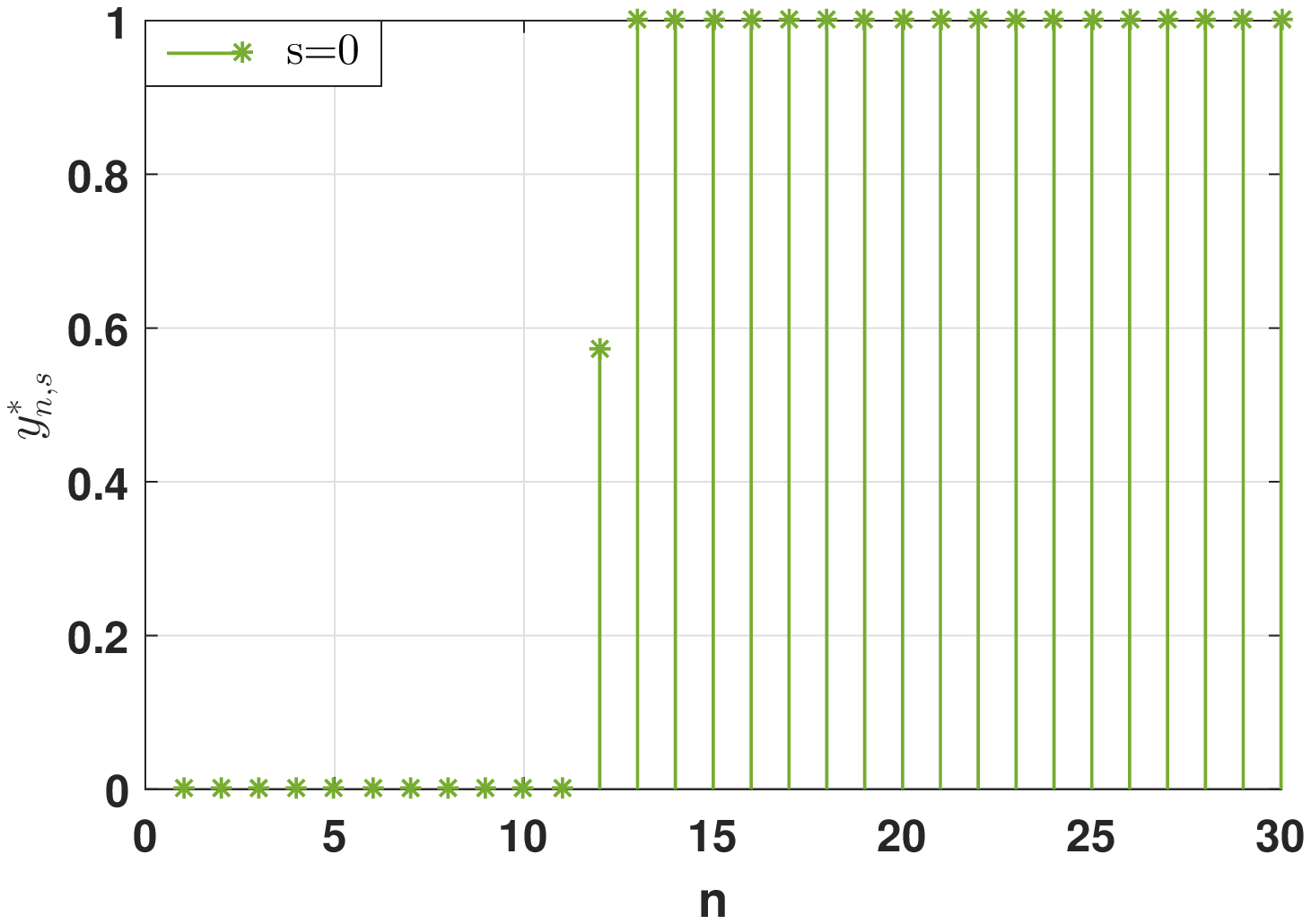}}}
              \subfigure[\small{$\gamma=1.5$.}]
      {\resizebox{4cm}{!}{\includegraphics{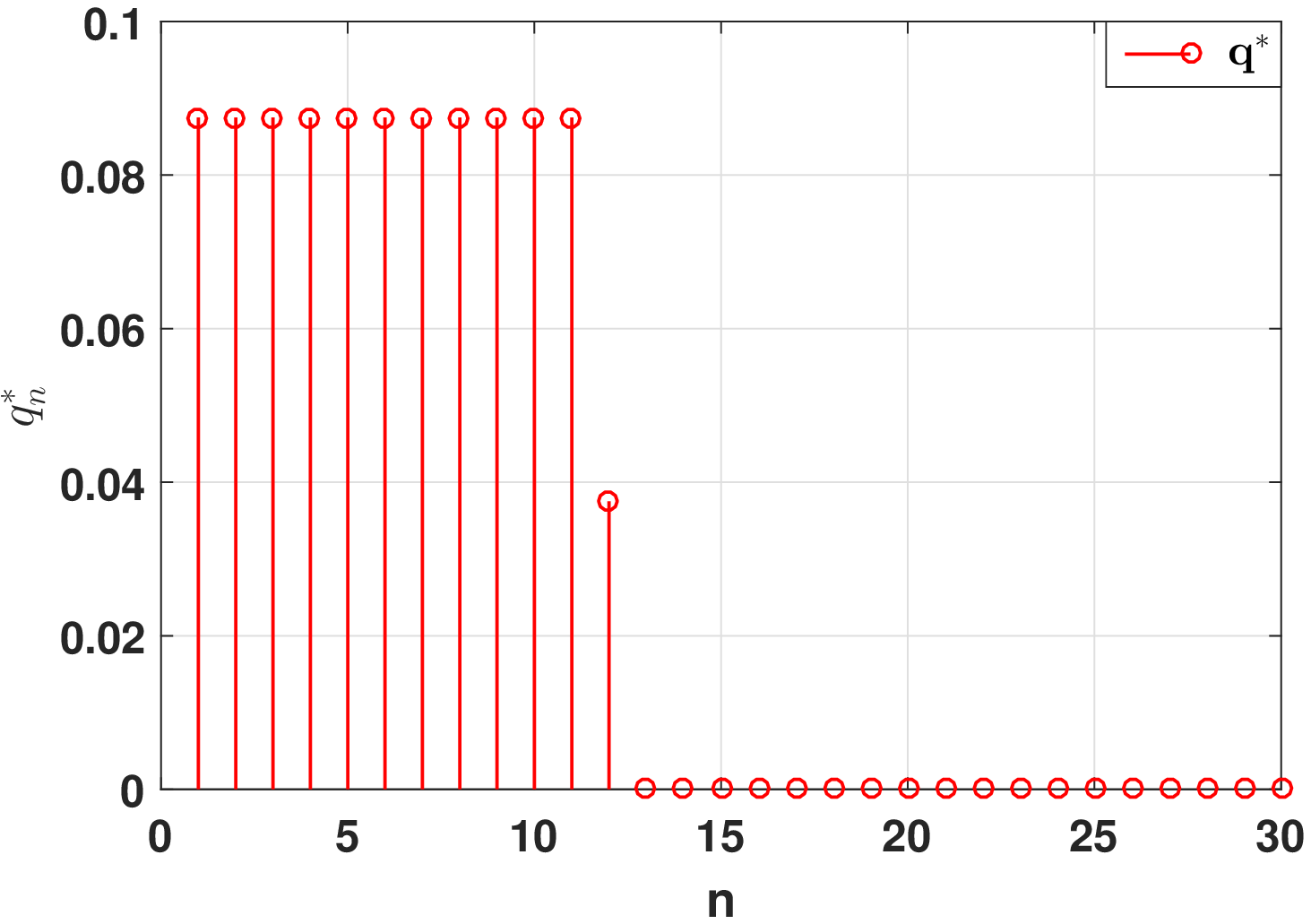}}}
\end{center}
         \caption{$\mathbf y^*$ and $\mathbf q^*$   of  the class of the centralized coded caching schemes  at $K=16$, $N=30$ and $M=5$. Note that in this case, $s$ can take values in the set $\{0,1,\ldots, 16\}$ and only the subset $\left\{s\in \{0,1,\ldots,16\}: \exists \ n \in \mathcal N \ s.t. \ y^*_{n,s}>0\right\}$ is plotted.
         }\label{fig:placement}
\end{figure}


\subsection{Special Properties of Optimized Parameter-based Scheme}
In this part, we demonstrate  special properties of   the centralized coded caching  scheme corresponding to   the  optimized  file partition parameter (referred to as  the   optimized  parameter-based scheme) using numerical results.
Let $\mathbf q^* \triangleq (q^*_n )_{n=1}^{N}$, where $q^*_n \triangleq \sum_{s=1}^{K} {K-1 \choose s-1} \frac{y^*_{n, s}}{M}$ denotes the fraction of the memory $M$ allocated to file $W_n$ at the optimized  file partition parameter.
Fig.~\ref{fig:placement} illustrates $\mathbf y^*$ and  $\mathbf q^*$ of  the   optimized  parameter-based scheme.
From  Fig.~\ref{fig:placement},  we can see that for all  $n \in \left\{1,2,\ldots, N-1\right\}$ and $s \in \{1,2, \cdots, K\}$,  $y^*_{n,s} \geq  y^*_{n+1,s}$, which verifies Theorem~\ref{Thm:popularity}.
Furthermore, we can see that files are classified  into different groups and for any file $n_1$ and file $n_2$ within the same group, we have $y^*_{n_1,s}=y^*_{n_2,s}$ for all $s \in \{0,1,\ldots, K\}$ and $q_{n_1}^*=q_{n_2}^*$, which   verifies Corollary~\ref{Cor:group} and validates the ``grouping" idea proposed in~\cite{NonuniformDemands}. From Fig.~\ref{fig:placement}, we can also see that at the optimized file partition parameter,   the average number of subfiles per file for $\gamma=1$ is $\frac{27}{30}\times\left({16 \choose 2}+{16 \choose 3}\right)+\frac{3}{30}\times\left({16 \choose 0}+{16 \choose 2}\right)=624.1$, and the average number of subfiles per file for $\gamma=1.5$ is $\frac{18}{30}\times{16 \choose 0}+\frac{12}{30}\times{16 \choose 7}=4576.6$, which are far smaller than $2^K=65536$.

\begin{figure}
\begin{center}
  \subfigure[\small{$\gamma=1.3$.}]
  {\resizebox{7cm}{!}{\includegraphics{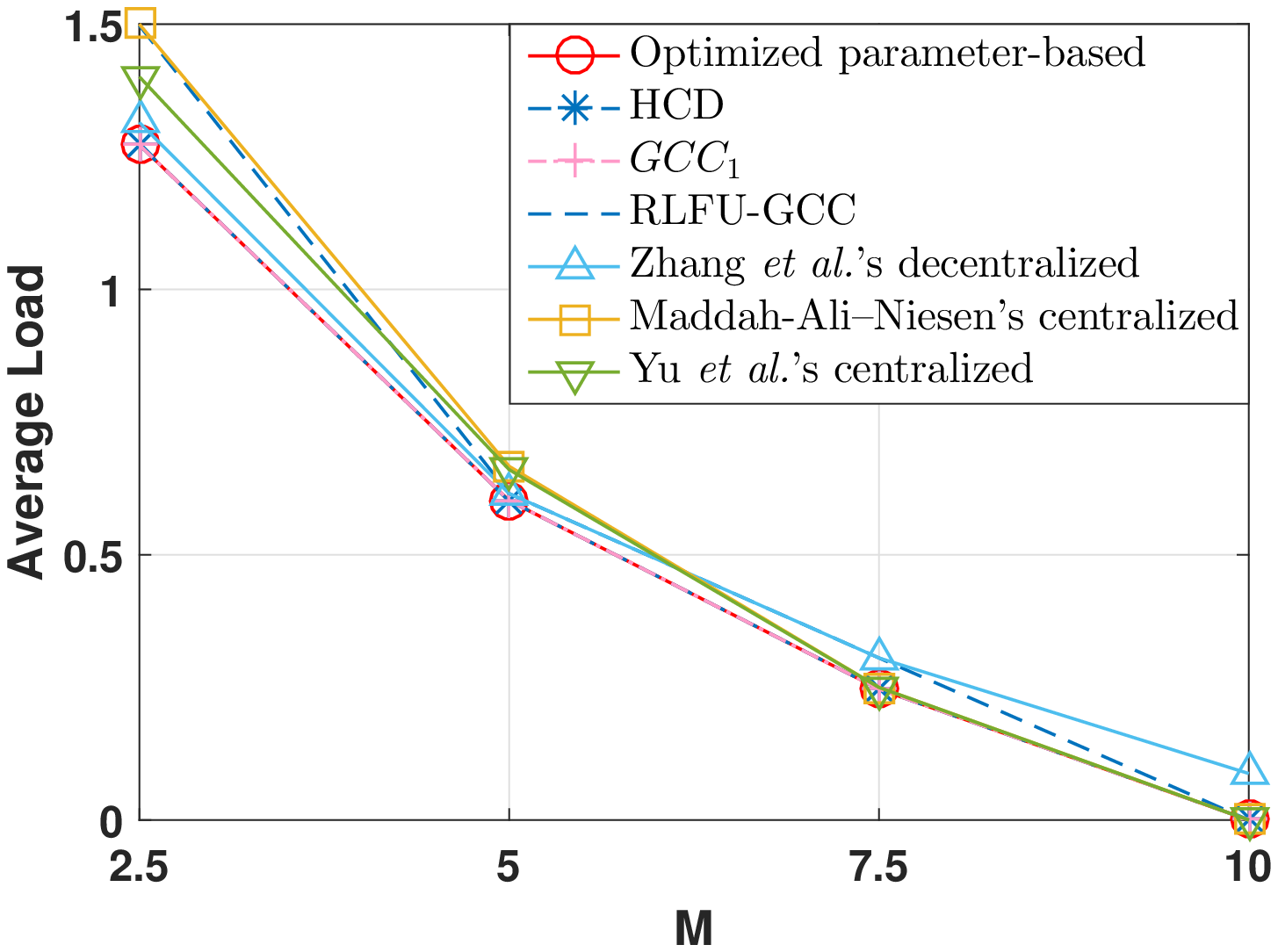}}}
  \quad
  \subfigure[\small{$\gamma=1.5$.}]
  {\resizebox{7cm}{!}{\includegraphics{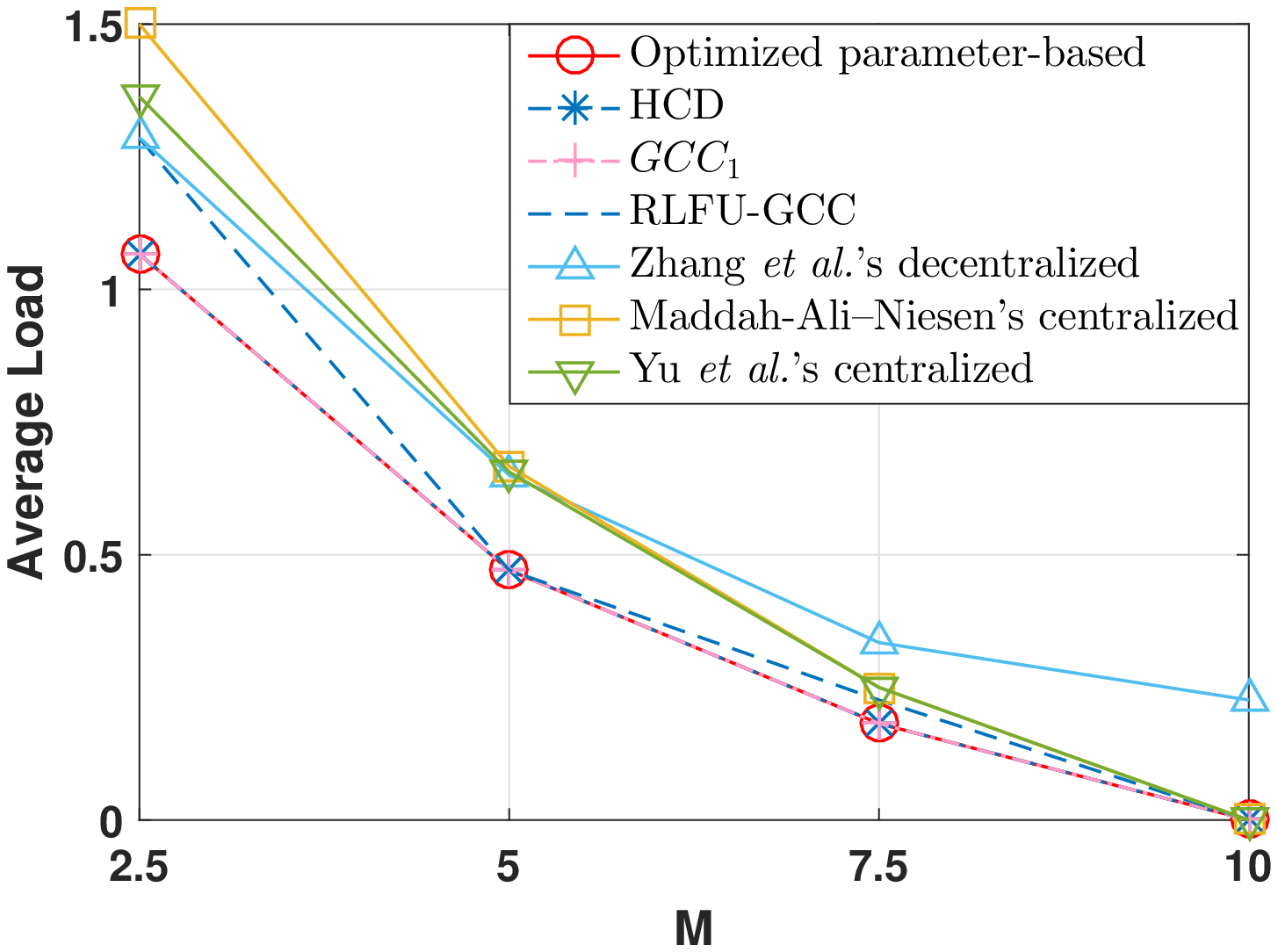}}}
  \end{center}
         \caption{Average load versus cache size $M$ when  $K=4$ and $N=10$.  Note that  Maddah-Ali--Niesen's centralized coded caching scheme mainly focuses on the cache size $M \in \{\frac{N}{K}, \frac{2N}{K}, \ldots, N\}$.  In the simulation, we consider the cache size $M \in \{\frac{N}{K}, \frac{2N}{K}, \ldots, N\}$ for all schemes, for purpose of comparison.
         }\label{fig:compare}
\end{figure}

\begin{figure}
\begin{center}
  {\resizebox{7cm}{!}{\includegraphics{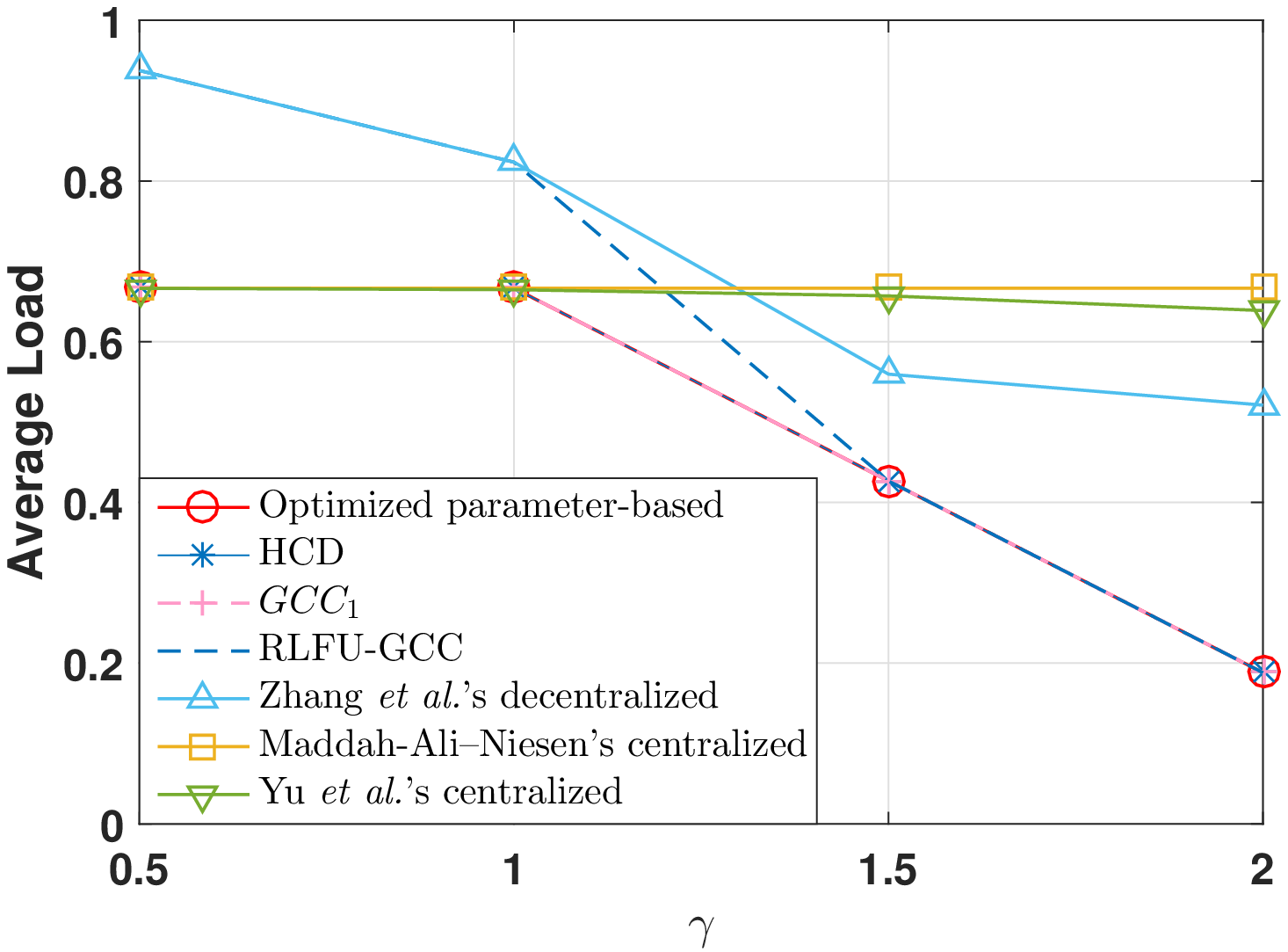}}}
         \caption{Average load versus cache size $\gamma$ when  $K=4$, $N=12$ and $M=6$.
         }\label{fig:simu_gamma}
\end{center}
\end{figure}

\subsection{Average Load Comparison}
In this part, we compare the average loads of the   optimized  parameter-based scheme,  the HCD procedure~\cite{HCD} at the  optimized  file partition parameter in Problem~\ref{Prob:simplify_2} (referred to as the HCD scheme here),  the $GCC_1$ procedure~\cite{ji2015order} at the  optimized  file partition parameter in Problem~\ref{Prob:simplify_2} (referred to as the $GCC_1$ scheme here),  the RLFU-GCC decentralized coded caching scheme (in the infinite file size regime)~\cite{ji2015order},  Zhang {\em et al.}'s decentralized coded caching  scheme (in the infinite file size regime)~\cite{Jinbei},  Maddah-Ali--Niesen's centralized coded caching scheme~\cite{AliFundamental} and Yu {\em et al.}'s centralized coded caching  scheme~\cite{YuQian}.

Fig.~\ref{fig:compare} illustrates the average loads of the above mentioned schemes versus the cache size $M$.
From Fig.~\ref{fig:compare}, we can see that the    optimized parameter-based scheme,  the $GCC_1$ scheme, and  the HCD scheme achieve the same average load $\overline{R}^*_{\rm avg}(K,N,M)$, which verifies Lemma~\ref{Lem:ji} and Lemma~\ref{Lem:HCD}.
Recall that the HCD scheme has higher complexity than the  optimized parameter-based scheme.
In addition, $\overline{R}^*_{\rm avg}(K,N,M)$ is no greater  than the average  loads of  Zhang {\em et al.}'s decentralized  coded caching scheme and Maddah-Ali--Niesen's centralized  coded caching scheme, which verifies  Statement~\ref{Sta:perforamnce_comparison}.
Moreover, $\overline{R}^*_{\rm avg}(K,N,M)$ is no greater than the average loads of  the RLFU-GCC decentralized coded caching scheme and  Yu {\em et al.}'s centralized coded caching  scheme at the parameters considered in the simulation. The reason that the   optimized  parameter-based scheme achieves better performance than the baseline schemes in~\cite{AliFundamental,ji2015order,Jinbei,YuQian}  is due to the advantage of the  optimized parameter-based scheme in exploiting  file popularity  for efficient content placement.

Fig.~\ref{fig:simu_gamma} illustrates the average loads of  the above mentioned schemes  versus  the Zipf exponent $\gamma$.
From Fig.~\ref{fig:simu_gamma}, we know that the average loads of the considered schemes, except  Maddah-Ali--Niesen's centralized coded caching scheme,  decrease as  $\gamma$  increases.
This is because Maddah-Ali--Niesen's centralized coded caching scheme is  designed for the worst-case and is  independent  of the file popularity distribution.
In addition, as  $\gamma$  increases,  the  average load gaps between the   optimized  parameter-based scheme and  Zhang {\em et al.}'s decentralized coded caching  scheme,  Maddah-Ali--Niesen's centralized coded caching scheme and Yu {\em et al.}'s centralized coded caching  scheme  increase. This   phenomenon indicates that the   optimized  parameter-based scheme can make better use of file popularity  for efficient content placement  when the  file popularity distribution  is highly non-uniform.

\begin{figure}
\begin{center}
  {\resizebox{7cm}{!}{\includegraphics{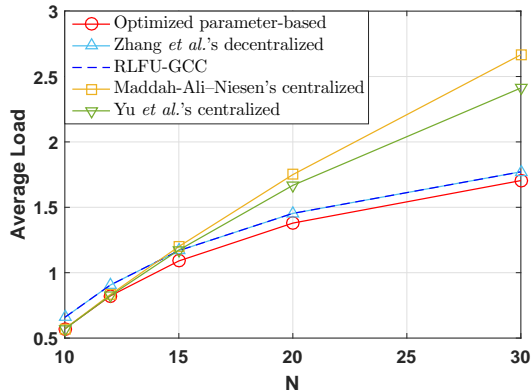}}}
         \caption{Average load versus  $N$ when  $\gamma=1.5$, $K=10$ and $M=6$.
         }\label{fig:simu_N}
\end{center}
\end{figure}

Fig.~\ref{fig:simu_N} illustrates the average loads of some of the above mentioned schemes  versus the number of files $N$.\footnote{Note that the HCD scheme and the $GCC_1$ scheme cannot be implemented using a  desktop when $K=10$  due to huge complexity.} From Fig.~\ref{fig:simu_N}, we see that the average load of each scheme increases with  $N$.  In addition, the  optimized  parameter-based scheme outperforms the other schemes in the considered regime of $N$.

\begin{figure}
\begin{center}
  \subfigure[\small{$\gamma=1$.}]
  {\resizebox{7cm}{!}{\includegraphics{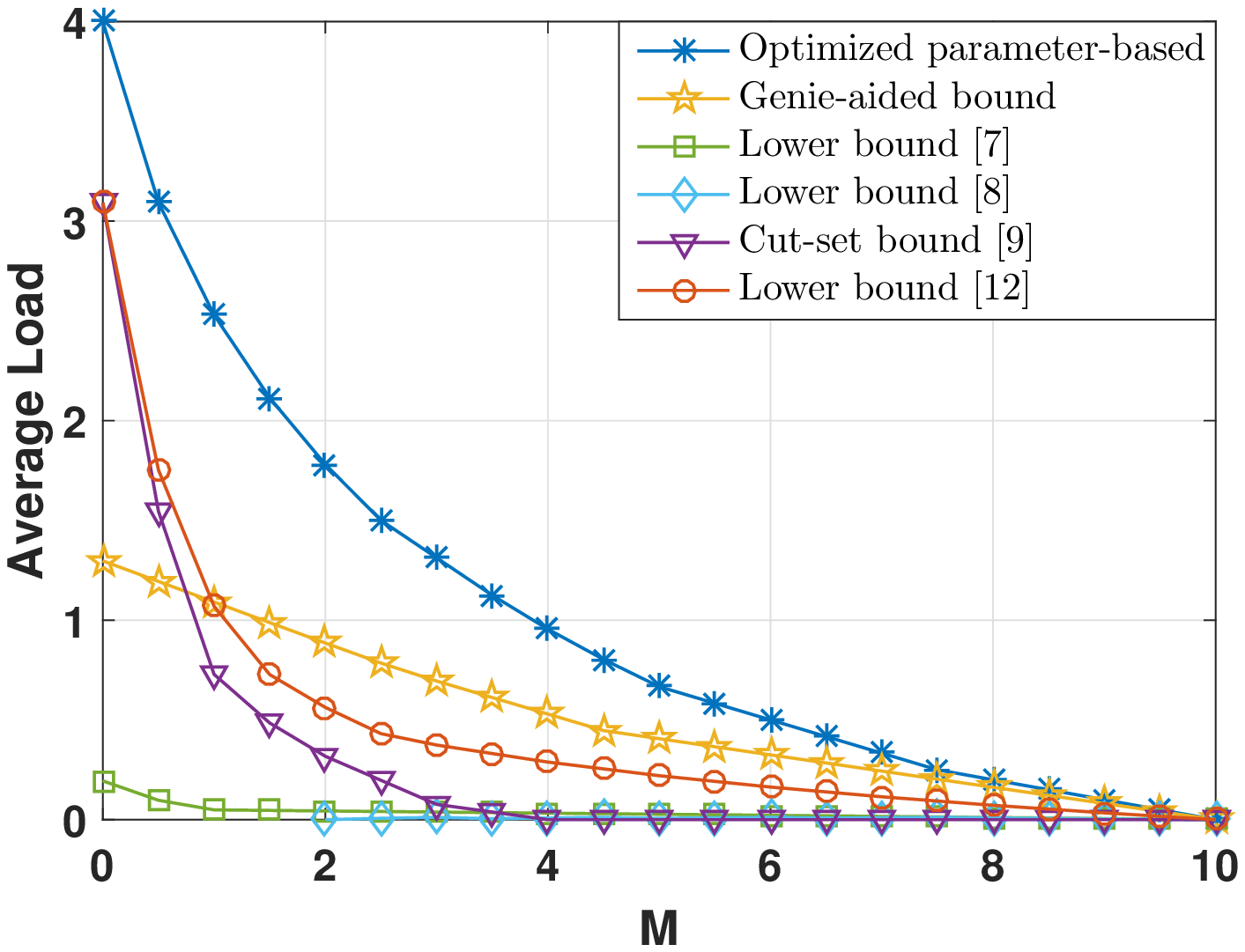}}}
  \quad
  \subfigure[\small{$\gamma=1.5$.}]
  {\resizebox{7cm}{!}{\includegraphics{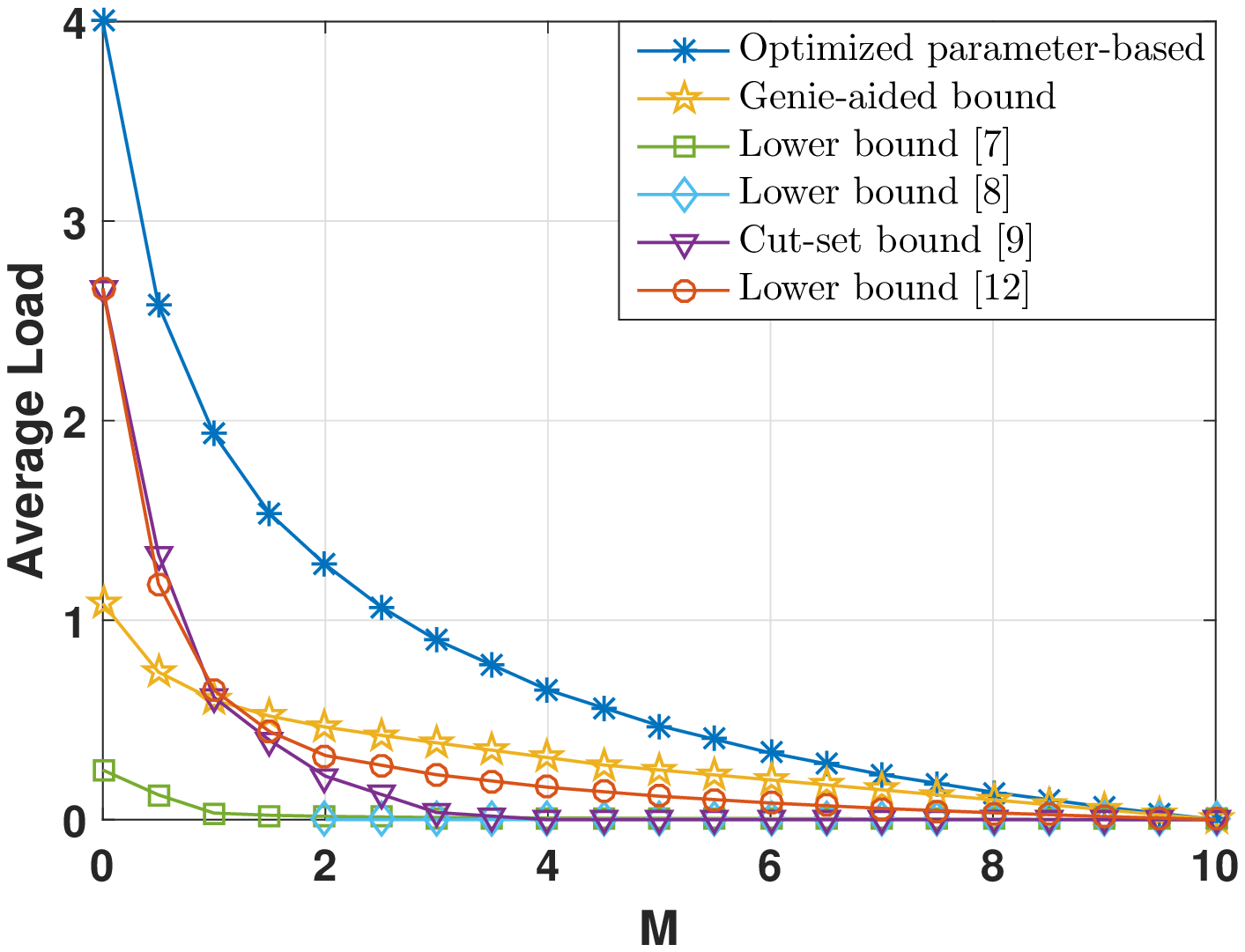}}}
  \end{center}
         \caption{$\overline{R}^*_{\rm avg}(K,N,M)$ and converse bounds  when  $K=4$ and $N=10$.
         }\label{fig:compare_lb}
\end{figure}

\subsection{Converse Bound Comparison}
In this part, we compare different information-theoretic converse bounds on the average load under  an arbitrary file popularity.
Fig.~\ref{fig:compare_lb} illustrates  the  average load of the  optimized  parameter-based scheme $\overline{R}^*_{\rm avg}(K,N,M)$, the genie-aided converse bound in \eqref{eqn:R_avg_lb}  and the converse bounds in~\cite{bound16,Jinbei,ji2015order,Sinong}. From Fig.~\ref{fig:compare_lb}, we can see that the genie-aided converse bound in \eqref{eqn:R_avg_lb} is tighter than the converse bounds in~\cite{ji2015order,Jinbei} for any cache size $M$. The genie-aided converse bound in \eqref{eqn:R_avg_lb} is tighter than the converse bounds in~\cite{bound16,Sinong}   when the cache size is modest or large, and is looser than the converse bounds in~\cite{bound16,Sinong}  when the cache size is small.

\section{Conclusion}
In this work,   we   considered a  class of  centralized  coded caching schemes    utilizing   general  uncoded placement  and  a specific coded delivery strategy, which  are specified   by a general  file partition parameter.
We  formulated the  coded caching design  optimization problem  to minimize the average load  over the considered class of schemes  by optimizing the file partition parameter  under an arbitrary file popularity.
We showed that the  optimization problem is convex, and the resulting optimal solution  generally improves upon known schemes.
Next, we analyzed structural properties of the optimization problem to obtain design insights and significantly reduce the complexity for obtaining an optimal solution.
Under the uniform file popularity,  we also obtained the  closed-form optimal  solution, which  corresponds to   Maddah-Ali--Niesen's centralized coded caching scheme.
Finally, we presented an information-theoretic converse bound on average load under an arbitrary file popularity, which was shown to improve on known bounds for arbitrary file popularity for some configurations of the system  parameters (in particular, for not too small cache memory).

This paper opens up several directions for future research.
For instance,   the class of the centralized coded caching schemes  can be extended to design efficient decentralized coded caching schemes to reduce the average load under an arbitrary file
popularity.
In addition,  the average load of the optimized  parameter-based scheme  may be further reduced by  using the improved delivery scheme of~\cite{YuQian}.
Finally, the parameter-based coded caching design approach can also be generalized to improve the performance of other coded caching schemes.

\section*{Appendix A: Proof of Lemma~\ref{Lem:ji}}
Consider a node $v$ in the conflict graph of the  $GCC_1$  procedure. If node $v$ corresponds to   subfile $W_{d_k,\mathcal{S}\setminus \{k\}}$ requested by user $k \in \mathcal S$, then  we have $\mu(v)=k$, $\eta(v)=\mathcal{S}\setminus \{k\}$, and $\{\mu(v), \eta(v)\}=\mathcal S$.
Note that under the placement procedure in Algorithm~\ref{alg:col}, $\{\mu(v_1), \eta(v_1)\}=\{\mu(v_2), \eta(v_2)\}$ iff $\mu(v_1) \in \eta(v_2)$ and  $\mu(v_2) \in \eta(v_1)$.
Thus, from the $GCC_1$  procedure, we know that assigning any two nodes $v_1$ and $v_2$ satisfying  $\{\mu(v_1), \eta(v_1)\}=\{\mu(v_2), \eta(v_2)\}$  the same color in the conflict graph corresponds to coding $W_{\mu(v_1),\eta(v_1)}$ and $W_{\mu(v_2),\eta(v_2)}$ together in  the delivery procedure in  Algorithm~\ref{alg:col}.
This means that all the nodes corresponding to  $\{\mu(v), \eta(v)\}$ can be assigned the same color as node $v$,  and the corresponding subfiles $W_{d_k,\mathcal{S} \setminus \{k\}}, k \in \mathcal S$ can be  coded together, as in the delivery procedure in  Algorithm~\ref{alg:col}.
Thus, we complete the proof of Lemma~\ref{Lem:ji}.
\section*{Appendix B: Proof of Theorem~\ref{Thm:symmetry}}
We prove  Theorem~\ref{Thm:symmetry} by considering the following two cases.

(i) Consider $s= K$. In this case, there exists only one subfile of type $K$, i.e.,  $W_{n,\mathcal K}$, and hence  Theorem~\ref{Thm:symmetry} holds obviously.

(ii) Consider any type  $s \in \{0, \ldots, K-1\}$ and  any feasible  file partition parameter $\mathbf{x}$.
Let $i_n$ denote the number of users requiring file $W_n$.
Let $\mathbf i \triangleq (i_1,i_2, \ldots, i_N)$ and
$$\mathcal I_s \triangleq \left\{\mathbf i \in \{0,1, \ldots, K\}^N: \sum_{n=1}^{N}i_n=s\right\}.$$
Let $\mathbf{d}_{\mathcal S} \triangleq (d_k)_{k \in \mathcal S} \in \mathcal N^{|\mathcal S|}$ denote the requests of the users in set $\mathcal S$.
For all $\mathbf i \in \mathcal I_s$ and
$\mathcal S \in \{\mathcal{S} \subseteq \mathcal{K}: |\mathcal{S}|=s\}$,  let
$\mathcal P_{\mathcal S, \mathbf i}\triangleq \left\{\mathbf{d}_{\mathcal S}   \in \mathcal N^{|\mathcal S|}: \sum_{k \in \mathcal S} \mathbf{1}[d_k=n]=i_n, n \in \mathcal N\right\}$,
$\Omega_{\mathcal S, \mathbf i} \triangleq \big\{(\mathcal S'_1, \mathcal S'_2, \ldots, \mathcal S'_N): \mathcal S'_n \subseteq \{\mathcal S' \subseteq \mathcal S: |\mathcal S'|=s-1\}, \cup_{n \in  \mathcal N} \mathcal S'_n =\{\mathcal S' \subseteq \mathcal S: |\mathcal S'|=s-1\}, |\mathcal S'_n |=i_n, n \in \mathcal N\big\}$.
By  \eqref{eqn:average_load_1}, we have
\begin{align}
&\overline{R}_{\rm avg}(K,N,M,\mathbf{x})= \sum_{s=1}^{K}\sum_{\mathcal{S} \in \{\mathcal{S} \subseteq \mathcal{K}: |\mathcal{S}|=s\}} \sum_{\mathbf{d} \in \mathcal{N}^K} \left(\prod_{k=1}^{K}p_{d_k}\right) \max_{k \in \mathcal{S}} x_{d_k,\mathcal{S}\setminus\{k\}} \nonumber\\
=& \sum_{s=1}^{K}\sum_{\mathcal{S}\in \{\mathcal{S} \subseteq \mathcal{K}: |\mathcal{S}|=s\}} \sum_{\mathbf{d}_{\mathcal S} \in \mathcal{N}^s} \left(\prod_{k \in \mathcal S}p_{d_k}\right)\max_{k \in \mathcal{S}} x_{d_k,\mathcal{S}\setminus\{k\}} \nonumber\\
=& \sum_{s=1}^{K}\sum_{\mathcal{S}\in \{\mathcal{S} \subseteq \mathcal{K}: |\mathcal{S}|=s\}}
\sum_{\mathbf i \in \mathcal I_s}\sum_{\mathbf{d}_{\mathcal S} \in \mathcal P_{\mathcal S, \mathbf i}} \left(\prod_{k \in \mathcal S}p_{d_k}\right)\max_{k \in \mathcal{S}} x_{d_k,\mathcal{S}\setminus\{k\}} \nonumber\\
=& \sum_{s=1}^{K}\sum_{\mathbf i \in \mathcal I_s}\left(\prod_{n \in \mathcal N}p_{n}^{i_n}\right) \sum_{\mathcal{S}\in \{\mathcal{S} \subseteq \mathcal{K}: |\mathcal{S}|=s\}}
\sum_{\mathbf{d}_{\mathcal S} \in \mathcal P_{\mathcal S, \mathbf i}} \max_{k \in \mathcal{S}} x_{d_k,\mathcal{S}\setminus\{k\}} \nonumber\\
=&\sum_{s=1}^{K}\sum_{\mathbf i \in \mathcal I_s}\left(\prod_{n \in \mathcal N}p_{n}^{i_n}\right) \sum_{\mathcal{S}\in \{\mathcal{S} \subseteq \mathcal{K}: |\mathcal{S}|=s\}}
\sum_{\mathbf{d}_{\mathcal S} \in \mathcal P_{\mathcal S, \mathbf i}}\max_{n \in \mathcal{N}:i_n>0} \max_{ k \in \{k \in \mathcal S: d_k=n\}} x_{n,\mathcal{S}\setminus\{k\}} \nonumber\\
=& \sum_{s=1}^{K} \sum_{\mathbf i \in \mathcal I_s}\left(\prod_{n \in \mathcal N}p_{n}^{i_n}\right)
\sum_{\mathcal{S}\in \{\mathcal{S} \subseteq \mathcal{K}: |\mathcal{S}|=s\}}
\sum_{(\mathcal S'_1, \mathcal S'_2, \ldots, \mathcal S'_N) \in \Omega_{\mathcal S, \mathbf i}} \max_{n \in \mathcal{N}:i_n>0} \max_{\mathcal{S}' \in \mathcal{S}'_n} x_{n,\mathcal{S}'}\nonumber\\
\overset{(a)}\geq &\sum_{s=1}^{K} \sum_{\mathbf i \in \mathcal I_s}\left(\prod_{n \in \mathcal N}p_{n}^{i_n}\right)
\sum_{\mathcal{S}\in \{\mathcal{S} \subseteq \mathcal{K}: |\mathcal{S}|=s\}}
\sum_{(\mathcal S'_1, \mathcal S'_2, \ldots, \mathcal S'_N) \in \Omega_{\mathcal S, \mathbf i}} \max_{\mathcal{S}' \in \mathcal{S}'_n}\frac{\sum_{\mathcal{S}' \in \mathcal{S}'_n} x_{n,\mathcal{S}'}}{i_n} \nonumber\\
=& \sum_{s=1}^{K} \sum_{\mathbf i \in \mathcal I_s}\left(\prod_{n \in \mathcal N}p_{n}^{i_n}\right)L(K,N,M,\mathbf{x}, \mathbf i),
\end{align}
where (a) is due to $\max_{\mathcal{S}' \in \mathcal{S}'_n} x_{n,\mathcal{S}'} \geq   \frac{\sum_{\mathcal{S}' \in \mathcal{S}'_n} x_{n,\mathcal{S}'}}{|\mathcal{S}'_n|}=\frac{\sum_{\mathcal{S}' \in \mathcal{S}'_n} x_{n,\mathcal{S}'}}{i_n}$, and $L(K,N,M,\mathbf{x}, \mathbf i)$ is given by
$$L(K,N,M,\mathbf{x}, \mathbf i) \triangleq \sum_{\mathcal{S}\in \{\mathcal{S} \subseteq \mathcal{K}: |\mathcal{S}|=s\}}
\sum_{(\mathcal S'_1, \mathcal S'_2, \ldots, \mathcal S'_N) \in \Omega_{\mathcal S, \mathbf i}} \max_{\mathcal{S}' \in \mathcal{S}'_n}\frac{\sum_{\mathcal{S}' \in \mathcal{S}'_n} x_{n,\mathcal{S}'}}{i_n}.$$
Next, we derive a lower bound of $L(K,N,M,\mathbf{x}, \mathbf i)$. Consider any $s \in \{1,2,\ldots, K\}$.
For all $\mathcal S' \in \{\mathcal S' \subseteq \mathcal K: |\mathcal S'|=s-1\}$,
the cardinality of $\{\mathcal{S} \subseteq \mathcal{K}: |\mathcal{S}|=s,\mathcal S' \subset\mathcal{S}\}$ is ${K-(s-1) \choose s-(s-1)}$.
Furthermore, for all $\mathcal{S}\in \{\mathcal{S} \subseteq \mathcal{K}: |\mathcal{S}|=s\}$,
the cardinality of $\{(\mathcal S'_1, \mathcal S'_2, \ldots, \mathcal S'_N) \in \Omega_{\mathcal S, \mathbf i}: \mathcal{S}' \in \mathcal S'_n\}$ is  ${s-1 \choose i_n-1}{s \choose i_1,\ldots, i_{n-1},i_{n+1}, \ldots, i_N}$.
Thus, we have
\begin{align}
&L(K,N,M,\mathbf{x}, \mathbf i) \nonumber\\
\overset{(b)}\geq &
\max_{n \in \mathcal{N}:i_n>0} \left\{{K-(s-1) \choose s-(s-1)} {s-1 \choose i_n-1} {s \choose i_1,\ldots, i_{n-1},i_{n+1}, \ldots, i_N} \frac{1}{i_n} \sum_{\mathcal S' \in \{\mathcal S' \subseteq \mathcal K: |\mathcal S'|=s-1\}}x_{n,\mathcal{S}'} \right\}\nonumber\\
=&\max_{n \in \mathcal{N}:i_n>0} \left\{{K-(s-1) \choose s-(s-1)} {s \choose i_1, i_2, \ldots, i_N}\frac{1}{s}\sum_{\mathcal S' \in \{\mathcal S' \subseteq \mathcal K: |\mathcal S'|=s-1\}}x_{n,\mathcal{S}'} \right\}\nonumber\\
=& \max_{n \in \mathcal{N}:i_n>0} \left\{ \frac{{K \choose s}}{{K \choose s-1}}{s \choose i_1, i_2, \ldots, i_N}\sum_{\mathcal S' \in \{\mathcal S' \subseteq\mathcal K: |\mathcal S'|=s-1\}}x_{n,\mathcal{S}'} \right\}\nonumber\\
=& {K \choose s} {s \choose i_1, i_2, \ldots, i_N} \max_{n \in \mathcal{N}:i_n>0} \left\{\frac{\sum_{\mathcal S' \in \{\mathcal S'\subseteq \mathcal K: |\mathcal S'|=s-1\}}x_{n,\mathcal{S}'}}{{K \choose s-1}} \right\}\nonumber\\
=& \sum_{\mathcal{S}\in \{\mathcal{S} \subseteq \mathcal{K}: |\mathcal{S}|=s\}}
\sum_{(\mathcal S'_1, \mathcal S'_2, \ldots, \mathcal S'_N) \in \Omega_{\mathcal S, \mathbf i}}
\max_{n \in \mathcal{N}:i_n>0} \left\{\frac{\sum_{\mathcal S' \in \{\mathcal S'\subseteq \mathcal K: |\mathcal S'|=s-1\}}x_{n,\mathcal{S}'}}{{K \choose s-1}} \right\}, \label{eqn:L_lb}
\end{align}
where (b) is due to $\max \{a_1, \ldots, a_N\}+\max \{b_1, \ldots, b_N\} \geq \max \{a_1+b_1, \ldots, a_N+b_N\}$. The equality  holds in  (b) when $x_{n, \mathcal{S}}=\frac{\sum_{\mathcal S' \subseteq \{ \mathcal S'\subseteq \mathcal K: |\mathcal S'|=s-1\}}x_{n,\mathcal{S}'}}{{K \choose s-1}}$ for all   $n \in \mathcal{N}$, $s \in \{1, \cdots, K\}$ and $\mathcal{S}\subseteq \{\mathcal{S} \subseteq \mathcal{K}: |\mathcal{S}|=s-1\}$.
Thus, we know  $x^*_{n, \mathcal{S}}=\frac{\sum_{\mathcal S' \subseteq \{ \mathcal S'\subseteq \mathcal K: |\mathcal S'|=s-1\}}x^*_{n,\mathcal{S}'}}{{K \choose s-1}}$ for all $n \in \mathcal{N}$, $s \in \{1, \cdots, K\}$ and $\mathcal{S}\subseteq \{\mathcal{S} \subseteq \mathcal{K}: |\mathcal{S}|=s-1\}$. 
Therefore, we complete the proof of Theorem~\ref{Thm:symmetry}.

\section*{Appendix C: Proof of Lemma~\ref{Lem:s symmetric}}
By \eqref{eqn:average_load_2}, we have
\begin{align}
&\widetilde{R}_{\rm avg}(K,N,M,\mathbf{y})=\sum_{s=1}^{K}\sum_{\mathcal{S}\in \{\mathcal{S} \subseteq \mathcal{K}: |\mathcal{S}|=s\}}\sum_{\mathbf{d} \in \mathcal{N}^K} \left(\prod_{k=1}^{K}p_{d_k} \right) \max_{k \in \mathcal{S}}y_{d_k,s-1}\nonumber \\
\overset{(a)}=&\sum_{s=1}^{K}{K \choose s}\sum_{\mathbf{d} \in \mathcal{N}^K} \left(\prod_{k=1}^{K}p_{d_k} \right)   \max_{k \in \{1,2,\cdots, s\}}y_{d_k,s-1}\nonumber \\
=&\sum_{s=1}^{K}{K \choose s}\sum_{(d_1, \ldots, d_s)  \in \mathcal{N}^s} \left(\prod_{k=1}^{s}p_{d_k}\right)  \max_{k \in \{1,2,\cdots, s\}}y_{d_k,s-1} \label{eqn: to} \\
\overset{(b)}=&\sum_{s=1}^{K}{K \choose s}\sum_{n=1}^{N}\sum_{(d_1, \ldots, d_s) \in \mathcal{D}_{n,s}} \left(\prod_{k=1}^{s}p_{d_k} \right)  \max_{k \in \{1,2,\cdots, s\}}y_{d_k,s-1},
\end{align}
where (a) is due to that for any $s \in \{1,2, \cdots, K\}$, the values of $\sum_{\mathbf{d} \in \mathcal{N}^K} \left(\prod_{k=1}^{K}p_{d_k}\right)  \max_{k \in \mathcal{S}}y_{d_k,s-1}$, $\mathcal{S} \subseteq \mathcal{K},|\mathcal{S}|=s$ are the same, and (b) is due to $\mathcal{N}^s=\cup_{n \in \mathcal N} \mathcal{D}_{n,s}$ and $\mathcal{D}_{n,s}\cap \mathcal{D}_{n',s}=\emptyset$ for all $n \neq n'$. Therefore, we complete the proof of Lemma~\ref{Lem:s symmetric}.
\section*{Appendix D: Proof of Theorem~\ref{Thm:popularity}}
Theorem~\ref{Thm:popularity} can be proved by proving  the following two statements.

{\em Statement (i):} For all  $n_1, n_2 \in \{1,2,\ldots, N\}$,
when $p_{n_1}=p_{n_2}$,  we have
\begin{align}
y^*_{n_1,s}  = y^*_{n_2,s}
\end{align}
for all $s \in \{1,2, \cdots, K\}$.

{\em Statement (ii):}¡¡For all  $n_1, n_2 \in \{1,2,\ldots, N\}$,
when $p_{n_1}>p_{n_2}$,  we have
\begin{align}
y^*_{n_1,s}  \geq y^*_{n_2,s} \label{eqn: s2}
\end{align}
for all $s \in \{1,2, \cdots, K\}$.

Next, we prove the above two statements, separately.
\subsection*{Proof of Statement (i)}
Consider any feasible  file partition parameter $\mathbf{y}$.
Let $\mathbf i_{-n_1,-n_2} \triangleq (i_n)_{n \in \mathcal N \setminus \{n_1, n_2\}}$,
$\mathcal I_{-n_1,-n_2, s-s_0} \triangleq \{\mathbf i_{-n_1,-n_2}: \sum_{n \in \mathcal N \setminus \{n_1, n_2\}}i_n=s-s_0\}$ and
$\mathcal I'_{n_1,n_2, s_0} \triangleq \{(i_{n_1}, i_{n_2}): i_{n_1}+i_{n_2}=s_0\}$.
By \eqref{eqn: to}, we have
\begin{align}
&\widetilde{R}_{\rm avg}(K,N,M,\mathbf{y}) \overset{(a)}=\sum_{s=2}^{K}{K \choose s}\sum_{(d_1, \ldots, d_s)  \in \mathcal{N}^s} \left(\prod_{k=1}^{s}p_{d_k}\right)  \max_{k \in \{1,2,\cdots, s\}}y_{d_k,s-1}
+K\left(1-\sum_{s=1}^{K} {K \choose s}\sum_{n=1}^{N}p_{n} y_{n, s}\right) \nonumber \\
=&\sum_{s=2}^{K}{K \choose s}\sum_{s_0=0}^{s}\sum_{\mathbf i_{-n_1,-n_2} \in \mathcal I_{-n_1,-n_2, s-s_0}}\sum_{(i_{n_1}, i_{n_2}) \in \mathcal I'_{n_1,n_2, s_0}}\sum_{(d_1, \ldots, d_s)  \in \mathcal P_{(1, \ldots, s), \mathbf i}}  p_{n_1}^{s_0}\left(\prod_{n \in \mathcal N \setminus \{n_1,n_2\}}p_{n}^{i_n}\right)   \nonumber \\
&\times\max_{k \in \{1,2,\cdots, s\}}y_{d_k,s-1}+K\left(1-\sum_{s=1}^{K} {K \choose s}\sum_{n=1}^{N}p_{n} y_{n, s}\right)\nonumber \\
=&\sum_{s=2}^{K}{K \choose s}\sum_{s_0=0}^{s}\sum_{\mathbf i_{-n_1,-n_2} \in \mathcal I_{-n_1,-n_2, s-s_0}} p_{n_1}^{s_0}\left(\prod_{n \in \mathcal N \setminus \{n_1,n_2\}}p_{n}^{i_n}\right) \sum_{(i_{n_1}, i_{n_2}) \in \mathcal I'_{n_1,n_2, s_0}}{s \choose i_1, i_2, \ldots, i_N}  \nonumber
\end{align}
\begin{align}
&\times \max_{k \in \{1,2,\cdots, s\}}y_{d_k,s-1}+K\left(1-\sum_{s=1}^{K} {K \choose s}\sum_{n=1}^{N}p_{n} y_{n, s}\right)\nonumber \\
=&\sum_{s=2}^{K}{K \choose s}\sum_{s_0=0}^{s}\sum_{\mathbf i_{-n_1,-n_2} \in \mathcal I_{-n_1,-n_2, s-s_0}}\frac{s!}{\prod_{n \in \mathcal N \setminus \{n_1,n_2\}}i_n!}p_{n_1}^{s_0}\left(\prod_{n \in \mathcal N \setminus \{n_1,n_2\}} p_{n}^{i_n}\right) \sum_{(i_{n_1}, i_{n_2}) \in \mathcal I'_{n_1,n_2, s_0} } {s_0 \choose i_{n_1}}   \nonumber \\
&\times \max_{n \in \mathcal N, i_n>0}y_{n,s-1}+K\left(1-\sum_{s=1}^{K} {K \choose s}\sum_{n=1}^{N}p_{n} y_{n, s}\right)\nonumber \\
\overset{(b)} \geq &\sum_{s=2}^{K}{K \choose s}\sum_{s_0=0}^{s}\sum_{\mathbf i_{-n_1,-n_2} \in \mathcal I_{-n_1,-n_2, s-s_0}}\frac{s!}{\prod_{n \in \mathcal N \setminus \{n_1,n_2\}}i_n!} p_{n_1}^{s_0}\left(\prod_{n \in \mathcal N \setminus \{n_1,n_2\}} p_{n}^{i_n}\right) \sum_{(i_{n_1}, i_{n_2}) \in \mathcal I'_{n_1,n_2, s_0} } {s_0 \choose i_{n_1}} \nonumber \\
&\times  \max \left\{\max_{n \in \mathcal N \setminus \{n_1,n_2\}:i_n>0}y_{n,s-1}, \frac{y_{n_1, s-1}+y_{n_2, s-1}}{2}\right\}+K\left(1-\sum_{s=1}^{K} {K \choose s}\sum_{n=1}^{N}p_{n} y_{n, s}\right),
\label{eqn:R_y_simp}
\end{align}
where (a) is due to  \eqref{eqn:X_sum_2},
and (b) is due to $\max \{a_1, \ldots, a_N\}+\max \{b_1, \ldots, b_N\} \geq \max \{a_1+b_1, \ldots, a_N+b_N\}$.
The equality holds in (b) when $y_{n_1, s-1}=y_{n_2, s-1}=\frac{y_{n_1, s-1}+y_{n_2, s-1}}{2}$ for all  $n_1, n_2 \in \{1,2,\ldots, N-1\}$ and   $s \in \{1,2, \cdots, K\}$.
Thus, $y^*_{n_1, s-1}=y^*_{n_2, s-1}=\frac{y^*_{n_1, s-1}+y^*_{n_2, s-1}}{2}$ for all $n_1, n_2 \in \{1,2,\ldots, N-1\}$ and   $s \in \{1,2, \cdots, K\}$.
Therefore, we complete the proof of Statement (i).
\subsection*{Proof of Statement (ii)}
First, we calculate the average loads under two related feasible file partition parameters.
Consider any $s_0 \in \{2,3, \ldots, K+1\}$ and  any feasible  file partition parameter $\mathbf{y}$.
Let $y_{(1), s_0-1} \geq y_{(2), s_0-1} \geq \ldots \geq y_{(N-1), s_0-1} \geq y_{(N), s_0-1}$ be the $y_{n, s_0-1} $'s arranged in decreasing order, so that $y_{(n), s_0-1}$ is the $n$-th largest.
Let $\mathcal{\widetilde{D}}_{(n),s_0}\triangleq \{(n),(n+1),\ldots, (N)\}^{s_0}\setminus \{(n+1),\ldots, (N)\}^{s_0}$.
For all $n \in \mathcal N$, we have
\begin{align}
\max_{k \in \{1,2,\cdots, s_0\}}y_{d_k,s_0-1}=y_{(n),s_0-1},\quad  (d_1, \ldots, d_{s_0}) \in \mathcal{\widetilde{D}}_{(n),s_0}. \label{eqn:max_to}
\end{align}
Then, by \eqref{eqn: to}, we have
\begin{align}
&\widetilde{R}_{\rm avg}(K,N,M,\mathbf{y})=\sum_{s=1}^{K}{K \choose s}\sum_{n=1}^{N}\sum_{(d_1, \ldots, d_s) \in \mathcal{\widetilde{D}}_{(n),s}} \left(\prod_{k=1}^{s}p_{d_k}\right)   \max_{k \in \{1,2,\cdots, s\}}y_{d_k,s-1} \nonumber \\
\overset{(a)}=&\sum_{s=1}^{K}{K \choose s}\sum_{n=1}^{N}\sum_{(d_1, \ldots, d_s) \in \mathcal{\widetilde{D}}_{(n),s}} \left(\prod_{k=1}^{s}p_{d_k}\right)  y_{i_{n},s-1}\nonumber \\
\overset{(b)}=&\sum_{s=2}^{K}{K \choose s}\sum_{n=1}^{N}\sum_{(d_1, \ldots, d_s) \in \mathcal{\widetilde{D}}_{(n),s}} \left(\prod_{k=1}^{s}p_{d_k}\right) y_{i_{n},s-1}
+K\left(1-\sum_{n=1}^{N}p_{i_{n}}\sum_{s=1}^{K} {K \choose s} y_{(n), s}\right), \label{eqn:R_y_o}
\end{align}
where (a) is due to \eqref{eqn:max_to} and (b) is due to \eqref{eqn:X_sum_2}. In addition,
let $n_0 \in\{1,2,\cdots,N-1\}$ denote the largest index such that  $p_{(n_0+1)}>p_{(n_0)}$.
By exchanging the values of $y_{(n_0), s_0-1}$ and $y_{(n_0+1), s_0-1}$, we can obtain another file partition parameter $\widehat{\mathbf{y}} \triangleq (\widehat{y}_{n,s})_{n \in \mathcal{N}, s \in \{0,1,\cdots, K\}}$, where
\begin{align}
\widehat{y}_{(n),s}=
\begin{cases}
y_{(n_0+1),s_0-1}, &n=n_0,   s=s_0-1\\
y_{(n_0),s_0-1}, &n=n_0+1,  s=s_0-1\\
1-\sum_{s \in \{1,2,\ldots, K\} \setminus \{s_0-1\}} {K \choose s} y_{(n_0+1), s}- {K \choose s_0-1} y_{(n_0), s_0-1}, &n=n_0+1,  s=0\\
1-\sum_{s \in \{1,2,\ldots, K\} \setminus \{s_0-1\}} {K \choose s} y_{(n_0), s}- {K \choose s_0-1} y_{(n_0+1), s_0-1}, &n=n_0,  s=0\\
y_{(n),s}, &\text{otherwise}.
\end{cases}\label{eqn:X_hat}
\end{align}
It is obvious that $\widehat{\mathbf{y}}$ is feasible.
For all $n \in \mathcal N \setminus\{n_0, n_0+1\}$, we have
\begin{align}
\max_{k \in \{1,2,\cdots, s_0\}}\widehat{y}_{d_k,s_0-1}=y_{(n),s_0-1}, \quad (d_1, \ldots, d_{s_0}) \in \mathcal{\widetilde{D}}_{(n),s_0}. \label{eqn:max_to_1}
\end{align}
Let $\mathcal{\widetilde{D}}'_{(n_0),s_0}\triangleq \{(n_0),(n_0+2),\ldots, (N)\}^{s_0}\setminus \{(n_0+2),\ldots, (N)\}^{s_0}$.
We have
\begin{align}
&\max_{k \in \{1,2,\cdots, s_0\}}\widehat{y}_{d_k,s_0-1}=y_{(n_0+1),s_0-1}, \quad (d_1, \ldots, d_{s_0}) \in \mathcal{\widetilde{D}}'_{(n_0),s_0}, \label{eqn:exchange_max_n_0_hat}
\end{align}
\begin{align}
&\max_{k \in \{1,2,\cdots, s_0\}}\widehat{y}_{d_k,s_0-1}=y_{(n_0),s_0-1}, \quad (d_1, \ldots, d_{s_0}) \in \mathcal{\widetilde{D}}_{(n_0),s_0}\setminus \mathcal{\widetilde{D}}'_{(n_0),s_0},  \label{eqn:exchange_max_n_2}
\end{align}
\begin{align}
&\max_{k \in \{1,2,\cdots, s_0\}}\widehat{y}_{d_k,s_0-1}=y_{(n_0),s_0-1}, \quad (d_1, \ldots, d_{s_0}) \in \mathcal{\widetilde{D}}_{(n_0+1),s_0}. \label{eqn:exchange_max_n_1_hat}
\end{align}
Then, by \eqref{eqn: to}, we have
\begin{align}
&\widetilde{R}_{\rm avg}(K,N,M,\widehat{\mathbf{y}})=\sum_{s=1}^{K}{K \choose s}\sum_{n=1}^{N}\sum_{(d_1, \ldots, d_s) \in \mathcal{\widetilde{D}}_{(n),s}} \left(\prod_{k=1}^{s}p_{d_k} \right)  \max_{k \in \{1,2,\cdots, s\}}\widehat{y}_{d_k,s-1}\nonumber \\
\overset{(c)}=&\sum_{s=2}^{K}{K \choose s}\sum_{n \in \mathcal N \setminus\{n_0, n_0+1\}}\sum_{(d_1, \ldots, d_s) \in \mathcal{\widetilde{D}}_{(n),s}} \left(\prod_{k=1}^{s}p_{d_k}\right)  y_{(n),s-1}\nonumber \\
+&\sum_{s \in \mathcal K \setminus\{1,s_0\}}{K \choose s}\sum_{n \in \{n_0, n_0+1\}}\sum_{(d_1, \ldots, d_s) \in \mathcal{\widetilde{D}}_{(n),s}} \left(\prod_{k=1}^{s}p_{d_k} \right)   y_{(n),s-1}\nonumber \\
+&{K \choose s_0}\left(\sum_{(d_1, \ldots, d_{s_0}) \in \mathcal{\widetilde{D}}'_{(n_0),s_0}} \left(\prod_{k=1}^{s}p_{d_k}\right)  y_{(n_0+1),s_0-1}
+\sum_{(d_1, \ldots, d_{s_0}) \in \mathcal{\widetilde{D}}_{(n_0),s_0}\setminus \mathcal{\widetilde{D}}'_{(n_0),s_0}\cup \mathcal{\widetilde{D}}_{(n_0+1),s_0}} \left(\prod_{k=1}^{s}p_{d_k} \right) y_{(n_0),s_0-1}\right)\nonumber
\end{align}
\begin{align}
+&K\left(1-\sum_{n \in \mathcal N \setminus\{n_0, n_0+1\}}p_{(n)}\sum_{s=1}^{K} {K \choose s} y_{(n), s}
-\sum_{s \in \mathcal K \setminus \{s_0-1\}} {K \choose s} \left(p_{(n_0)}y_{(n_0), s}+p_{(n_0+1)}y_{(n_0+1), s}\right)\right)\nonumber \\
-& K{K \choose s_0-1} \left(p_{(n_0)}y_{(n_0+1), s_0-1}+p_{(n_0+1)}y_{(n_0), s_0-1}\right),
\label{eqn:R_y_v}
\end{align}
where (c) is due to \eqref{eqn:X_sum_2}, \eqref{eqn:max_to_1}--\eqref{eqn:exchange_max_n_1_hat}.

Next, we prove $\widetilde{R}_{\rm avg}(K,N,M,\mathbf{y}) \geq  \widetilde{R}_{\rm avg}(K,N,M,\widehat{\mathbf{y}})$.
By \eqref{eqn:R_y_o} and \eqref{eqn:R_y_v}, we have
\begin{align}
&\widetilde{R}_{\rm avg}(K,N,M,\mathbf{y})-\widetilde{R}_{\rm avg}(K,N,M,\widehat{\mathbf{y}})\nonumber \\
=&\left({K \choose s_0}\left(\sum_{(d_1, \ldots, d_{s_0}) \in \mathcal{\widetilde{D}}'_{(n_0),s_0}} \left(\prod_{k=1}^{s_0}p_{d_k}\right)
-\sum_{(d_1, \ldots, d_{s_0}) \in \mathcal{\widetilde{D}}_{(n_0+1),s_0}} \left(\prod_{k=1}^{s_0}p_{d_k}\right)\right)+K{K \choose s_0-1}\left(p_{(n_0+1)}-p_{(n_0)}\right)\right)\nonumber \\
&\times\left(y_{(n_0),s_0-1}-y_{(n_0+1),s_0-1}\right)\nonumber \\
=&f(s_0){K \choose s_0-1}\left(y_{(n_0),s_0-1}-y_{(n_0+1),s_0-1}\right), \label{eqn:load_comp}
\end{align}
where
\begin{align}
f(s)\triangleq \frac{K-s+1}{ s}\left(\left(p_{(n_0)}+\sum_{n'=n_0+2}^{N}p_{(n')}\right)^{s}-\left(\sum_{n'=n_0+1}^{N}p_{(n')}\right)^{s}\right)+K\left(p_{(n_0+1)}-p_{(n_0)}\right). \label{eqn:f_s}
\end{align}
To prove  $\widetilde{R}_{\rm avg}(K,N,M,\mathbf{y}) \geq  \widetilde{R}_{\rm avg}(K,N,M,\widehat{\mathbf{y}})$, it is sufficient to show  $f(s_0)>0$, for all  $s_0 \in \{2, \ldots, K\}$.
By \eqref{eqn:f_s}, we have $f'(s)=g(\alpha)-g(\beta)$, where $g(x) \triangleq \left(\frac{K-s+1}{s}\ln x-\frac{K+1}{s^2}\right)x^s$, $\alpha \triangleq p_{(n_0)}+\sum_{n'=n_0+2}^{N}p_{(n')}$ and $\beta \triangleq \sum_{n'=n_0+1}^{N}p_{(n')}$.
For all $s \in \{1,2, \ldots,K\}$ and $x \in (0,1)$, we have  $g'(x)= \left(\left(K-s+1\right)\ln x -1\right)x^{s-1}<0$.
By noting that $0<\alpha<\beta<1$,  we have $f'(s)=g(\alpha)-g(\beta) >0$, implying that $f(s) > f(1)=0$ for all $s \in \{2, \ldots, K\}$.
Thus, by \eqref{eqn:load_comp},  we can show $\widetilde{R}_{\rm avg}(K,N,M,\mathbf{y})  \geq  \widetilde{R}_{\rm avg}(K,N,M,\widehat{\mathbf{y}})$.

From the above discussion, we know that when $p_{(n_0+1)}>p_{(n_0)}$, by exchanging the values of $y_{(n_0),s_0-1}$ and $y_{(n_0+1),s_0-1}$, we can always reduce the average load. Thus, for the optimized solution  $\mathbf{y}^*$, there does not exist any $n_0 \in \{1,2,\ldots, N-1\}$ such that $p_{(n_0+1)}>p_{(n_0)}$. In other words, for all  $n_1, n_2 \in \{1,2,\ldots, N-1\}$ satisfying $p_{n_1}>p_{n_2}$, we have $y^*_{n_1,s_0-1}  \geq y^*_{n_2,s_0-1}$ for all $s_0 \in \{2,3, \cdots,  K+1\}$.
Therefore, we complete the proof of Statement~(ii).

\section*{Appendix E: Proof of Corollary~\ref{Cor:group}}
We prove Corollary~\ref{Cor:group} by proving the sufficiency and necessity.
First, we prove the sufficiency. If  $y^*_{n_1,s}=y^*_{n_2,s}$  for all $s \in \{0,1,\ldots, K\}$, obviously we have
$$\sum_{s=1}^{K} {K-1 \choose s-1} y^*_{n_1, s}=\sum_{s=1}^{K} {K-1 \choose s-1} y^*_{n_2, s}.$$
Next, we prove the necessity. Without loss of generality, we suppose $p_{n_1} \geq p_{n_2}$. By Theorem~\ref{Thm:popularity}, we have
\begin{align}
y^*_{n_1,s} \geq y^*_{n_2,s}, \quad \forall s \in \{1,\ldots, K\}. \label{eqn:cor}
\end{align}
If $$\sum_{s=1}^{K} {K-1 \choose s-1} y^*_{n_1, s}=\sum_{s=1}^{K} {K-1 \choose s-1} y^*_{n_2, s},$$
by \eqref{eqn:cor}, we have $y^*_{n_1,s}=y^*_{n_2,s}$  for all $s \in \{1,\ldots, K\}$. Based on this, by~\eqref{eqn:X_sum_2}, we have $y^*_{n_1,0}=y^*_{n_2,0}$.
Therefore, we complete the proof of Corollary~\ref{Cor:group}.

\section*{Appendix F: Proof of Lemma~\ref{Lem:simplification pop}}
By \eqref{eqn:additional_cons}, for any $(d_1, \ldots, d_{s}) \in \mathcal{D}_{n,s}$, $n \in\{1,2,\cdots,N-1\}$ and $s \in \mathcal{K}$, we have
\begin{align}
\max_{k \in \{1,2,\cdots, s\}}y_{d_k,s-1}=y_{n,s-1}, \quad s \in \mathcal{K}.\label{eqn:additional_cons_cor}
\end{align}
By \eqref{eqn:average_load_3} and \eqref{eqn:additional_cons_cor}, we have
\begin{align}
&\widetilde{R}_{\rm avg}(K,N,M,\mathbf{y}) \overset{(a)}=\sum_{s=1}^{K}{K \choose s}  \sum_{n=1}^{N} y_{n,s-1}  \sum_{(d_1, \ldots, d_{s})  \in \mathcal{D}_{n,s}}  \left(\prod_{k=1}^{s}p_{d_k}\right)  \nonumber \\
\overset{(b)}=&\sum_{s=1}^{K}{K \choose s}  \sum_{n=1}^{N} y_{n,s-1}  \left(\sum_{(d_1, \ldots, d_{s})  \in \{n,n+1,\ldots,N\}^s}  \left(\prod_{k=1}^{s}p_{d_k}\right) - \sum_{(d_1, \ldots, d_{s})  \in \{n+1,n+2,\ldots,N\}^s } \left(\prod_{k=1}^{s}p_{d_k}\right)\right) \nonumber \\
=&\sum_{s=1}^{K}{K \choose s}\sum_{n=1}^{N}  y_{n,s-1} \left(\left(\sum_{n'=n}^{N}p_{n'}\right)^s-\left(\sum_{n'=n+1}^{N}p_{n'}\right)^s\right) ,
\end{align}
where (a) is due to \eqref{eqn:additional_cons_cor} and (b) is due to the definition of $\mathcal{D}_{n,s}$.
Therefore, we complete the proof of Lemma~\ref{Lem:simplification pop}.

\section*{Appendix G: Proof of Lemma~\ref{Lem:HCD}}
Consider any subset  of users  $\mathcal S\subset \mathcal{K}$ and $\mathbf{D}=\mathbf{d} \triangleq (d_k)_{k \in \mathcal K}$.
Define   $\mathcal{S}_m \triangleq \{k \in \mathcal{S}: d_k= \underset{l \in \mathcal{S}} \min d_l \}$.
By \eqref{eqn:symmetry} and \eqref{eqn:additional_cons}, we know that for any $k \in \mathcal S \setminus \mathcal{S}_m$ and any $k_m \in  \mathcal{S}_m$,  subfile $W_{d_k, \mathcal S \setminus \{k\}}$  has  smaller length  than  subfile $W_{d_{k_m}, \mathcal S \setminus \{k_m\}}$, and should be padded with bits from some subfiles  $W_{d_k, \mathcal S' \setminus \{k\}}$, $\mathcal S \subset \mathcal S' \subseteq \mathcal{K}$ in the HCD procedure.
Consider any subset of users $\mathcal{S}'$ such that $\mathcal S \subset \mathcal S' \subseteq \mathcal{K}$.
Define   $\mathcal{S}'_m \triangleq \{k \in \mathcal{S}': d_k= \underset{l \in \mathcal{S}'} \min d_l \}$.
By the definitions of $\mathcal{S}_m $ and  $\mathcal{S}'_m $, we have
\begin{align}
d_{k'_m} < d_k, \quad \forall k'_m \in  \mathcal{S}'_m, \ k \in \mathcal S \setminus \mathcal{S}_m, \label{eqn:d_p_k_m_less_k}
\end{align}
$\left(\mathcal S \setminus \mathcal{S}_m\right) \cap \mathcal{S}'_m =\emptyset$ and $\left(\mathcal S \setminus \mathcal{S}_m\right) \cup \mathcal{S}'_m  \subset \mathcal S'$.
Consider any $k \in \mathcal S \setminus \mathcal{S}_m$ and any $k'_m \in \mathcal{S}'_m$.
By \eqref{eqn:symmetry}, \eqref{eqn:additional_cons} and \eqref{eqn:d_p_k_m_less_k}, we know that the size of subfile $W_{d_k, \mathcal S' \setminus \{k\}}$   in   coded multicast message $\oplus_{k \in \mathcal S'} W_{d_{k},\mathcal S' \setminus \{k\}}$   is smaller than that of the longest subfile $W_{d_{k'_m}, \mathcal S' \setminus \{k'_m\}}$   in this coded multicast message.
Thus, the appending method  in the  HCD procedure does not change the size of $\oplus_{k \in \mathcal S'} W_{d_{k},\mathcal S' \setminus \{k\}}$.
Therefore, we complete the proof of Lemma~\ref{Lem:HCD}.

\section*{Appendix H: Proof of Lemma~\ref{Lem:Ali}}
The Lagrangian of Problem~\ref{Prob:equivalent_3}  is given by
\begin{align}
L(\mathbf{z}, \boldsymbol \eta, \theta, \nu)=&\sum_{s=0}^{K}{K \choose s}\frac{K-s}{s+1} z_{s}+ \eta_{s} \left(-z_{s}\right)+\theta  \left(\sum_{s=0}^{K} {K \choose s}s z_{s}-\frac{KM}{N}\right)+ \nu \left(1-\sum_{s=0}^{K} {K \choose s} z_{s}\right), \nonumber
\end{align}
where $\eta_{s} \geq 0$ is the Lagrange multiplier associated with~\eqref{eqn:X_range_3}, $\nu$ is the Lagrange multiplier associated with~\eqref{eqn:X_sum_3}, $\theta$ is the Lagrange multiplier associated with~\eqref{eqn:memory_constraint_3} and $\boldsymbol \eta \triangleq (\eta_{s})_{s \in \{0,1,\ldots,K\}}$. Thus, we have
\begin{align}
\frac{\partial L}{\partial \eta_{s}}(\mathbf{z}, \boldsymbol \eta,  \theta, \nu)={K \choose s}\frac{K-s}{s+1}-\eta_{s}+\theta s{K \choose s}-\nu {K \choose s}.
\end{align}
Since strong duality holds, primal optimal $\mathbf{z}^*$ and dual optimal $\boldsymbol \eta^*, \nu^*, \theta^*$ satisfy KKT conditions, i.e., (i) primal constraints: \eqref{eqn:X_range_3}, \eqref{eqn:X_sum_3}, \eqref{eqn:memory_constraint_3}, (ii) dual constraints: (a) $\eta_{s} \geq 0$    for all $s \in \{0,1,\ldots,K\}$ and (b) $\theta \geq 0$,
(iii) complementary slackness: (a) $\eta_{s} \left(-z_{s}\right)=0$  for all $s \in \{0,1,\ldots,K\}$ and (b) $\theta  \left(\sum_{s=0}^{K} {K \choose s}s z_{s}-\frac{KM}{N}\right)=0$,  and (iv) ${K \choose s}\frac{K-s}{s+1}-\eta_{s}+\theta s{K \choose s} -\nu {K \choose s}=0$  for all $s \in \{0,1,\ldots,K\}$.
By (ii.a) and (iv), we know that for all $s \in \{0,1,\ldots,K\}$,  $\eta^*_{s}= {K \choose s}\left(\frac{K-s}{s+1}+\theta^* s -\nu^* \right) \geq 0$, implying
\begin{align}
h(s) \triangleq \theta^* s^2+(\theta^*-\nu^*-1)s+K-\nu^* \geq 0.\label{eqn:s_geq_0}
\end{align}
Furthermore, for all   $s \in \{0,1,\ldots,K\}$, when $z^*_s>0$,  by  (iii.a) and (iv),  we have $\eta^*_{s}= {K \choose s}\left(\frac{K-s}{s+1}+\theta^* s-\nu^*\right)=0$,  implying
\begin{align}
h(s) =0. \label{eqn:s_eq_0}
\end{align}
That is, for any $s \in \{0,1,\ldots,K\}$, when \eqref{eqn:s_eq_0} does not hold, $z^*_s=0$.
Since \eqref{eqn:s_eq_0} has at most two different roots, there are at most  two    $s \in \{0,1,\ldots,K\}$ such that  $z^*_s>0$.
In addition, by  \eqref{eqn:X_sum_3}, we know that there exists  at least one $s \in \{0,1,\ldots,K\}$ such that  $z^*_s>0$.
Thus, there exist one or two  $s \in \{0,1,\ldots,K\}$ such that  $z^*_s>0$.
In the following, consider  two possible cases, i.e., $\theta^*=0$ and $\theta^*>0$.
\begin{itemize}
\item
Consider  $\theta^*=0$.
By \eqref{eqn:s_eq_0},  we have  $z^*_s>0$ for  $s=\frac{K-\nu^*}{\nu^* +1}$, implying $\frac{K-\nu^*}{\nu^* +1} \leq  K$, and $z^*_s=0$ for $s \in \{0,1,\cdots, K\} \setminus \left\{\frac{K-\nu^*}{\nu^* +1}\right\}$.
By  \eqref{eqn:s_geq_0}, we have $s \leq \frac{K-\nu^*}{\nu^* +1}$ for all $s \in \{0,1,\cdots, K\}$, implying $K \leq \frac{K-\nu^*}{\nu^* +1}$.
Thus, we have $\frac{K-\nu^*}{\nu^* +1} = K$,  implying $\nu^*=0$,  $z^*_s>0$ for  $s=K$ and $z^*_s=0$ for $s \in \{0,1,\cdots, K-1\}$.
Then, by \eqref{eqn:X_sum_3}, we have
\begin{align}
z^*_{s}=
\begin{cases}
1, & s=K\\
0, &s \in \{0,1,\cdots, K-1\}.
\end{cases}\label{eqn:theta_eq_0}
\end{align}
By $\theta^*=0$, $\nu^*=0$ and (iv), we have $\eta^*_{s}= {K \choose s} \frac{K-s}{s+1}$, $s \in \{0,1,\cdots, K\}$.
Note that when $M \in \left\{0, \frac{N}{K},  \ldots, \frac{(K-1)N}{K}\right\}$, $\mathbf z^*$ given in \eqref{eqn:theta_eq_0} does not satisfy \eqref{eqn:memory_constraint_3}.
When $M=N$, $\mathbf z^*$ given in \eqref{eqn:theta_eq_0},  $\theta^*=0$, $\nu^*=0$ and $\eta^*_{s}= {K \choose s} \frac{K-s}{s+1}$, $s \in \{0,1,\cdots, K\}$  satisfy the KKT conditions in (i)-(iv).
Thus,  $\mathbf z^*$ given in \eqref{eqn:theta_eq_0} is the unique  optimal solution  when  $M= N$. Note that  when $M=N$, \eqref{eqn:uniform_caching} reduces to~\eqref{eqn:theta_eq_0}.
\item
Consider  $\theta^*>0$.
By  (iii.b), we have
\begin{align}
\sum_{s=0}^{K} {K \choose s}s z_{s}=\frac{KM}{N}. \label{eqn:eq}
\end{align}
First, we  prove that  there is only one $s \in \{0,1,\ldots,K\}$ such that $z^*_s>0$ by contradiction.
Suppose there exist  two $s_1, s_2 \in \{0,1,\ldots,K\}$, $s_1 \neq s_2$,     such that $z^*_{s_1}>0$ and $z^*_{s_2}>0$.
Then, $s_1$ and $s_2$ are two different roots of \eqref{eqn:s_eq_0}, i.e.,  $h(s_1)=h(s_2)=0$, and \eqref{eqn:X_sum_3} implies
\begin{align}
{K \choose s_1} z^*_{s_1}+{K \choose s_2}z^*_{s_2}=1.\label{eqn:contrad_s_eq_2_2}
\end{align}
In addition, by \eqref{eqn:eq}, we have
\begin{align}
{K \choose s_1}s_1 z^*_{s_1}+{K \choose s_2}s_2 z^*_{s_2}=\frac{KM}{N}. \label{eqn:contrad_s_eq_2_1}
\end{align}
When $M=N$, by \eqref{eqn:contrad_s_eq_2_2} and \eqref{eqn:contrad_s_eq_2_1}, we have $s_1=s_2=K$, which contradicts $s_1 \neq s_2$.
When $M<N$, without loss of generality, we suppose $s_2>s_1$. By \eqref{eqn:contrad_s_eq_2_2} and \eqref{eqn:contrad_s_eq_2_1}, we have
\begin{align}
z^*_{s_1}=\frac{\frac{KM}{N}-s_1}{{K \choose s_1}\left(s_2-s_1\right)}, \label{eqn:z_1}
\end{align}
\begin{align}z^*_{s_2}=\frac{s_2-\frac{KM}{N}}{{K \choose s_2}\left(s_2-s_1\right)}. \label{eqn:z_2}
\end{align}
 By $z^*_{s_1}>0$,  $z^*_{s_2}>0$, \eqref{eqn:z_1}  and \eqref{eqn:z_2}, we have
\begin{align}
s_1<\frac{KM}{N}<s_2. \label{eqn:contrad_s_1_s_2}
\end{align}
Note that $\frac{KM}{N} \in \{0,1, \ldots, K\}$, as $M \in \left\{0, \frac{N}{K}, \frac{2N}{K}, \ldots, N\right\}$. In addition, recall that $s \in \{0,1, \ldots, K\}$.
Thus, when $K=1$,   \eqref{eqn:contrad_s_1_s_2} contradicts $s_1, s_2 \in \{0,1\}$.
When $K \in \{2,3, \ldots\}$, since $\theta^*>0$ and $h(s_1)=h(s_2)=0$, by \eqref{eqn:contrad_s_1_s_2}, we know that
\begin{align}
h\left(\frac{KM}{N}\right)<0, \label{eqn:h_less_0}
\end{align}
which contradicts   \eqref{eqn:s_geq_0}.
Therefore, we can show that if $\theta^*>0$,   there is only one $s \in \{0,1,\ldots,K\}$ such that $z^*_s>0$.
Then, by  \eqref{eqn:X_sum_3} and \eqref{eqn:eq}, we can obtain  \eqref{eqn:uniform_caching}.
Next, we prove that \eqref{eqn:uniform_caching} is the   optimal solution for any    $M \in \left\{0, \frac{N}{K}, \frac{2N}{K}, \ldots, N\right\}$.
When $M=0$,   $\mathbf z^*$ given in \eqref{eqn:uniform_caching},   any $\theta^* \in (0,K+1]$,  $\nu ^*=K$ and    $\eta^*_{s}= {K \choose s}\left(\frac{K-s}{s+1}+\theta^* s -K\right)$, $s \in \{0,1,\ldots,K\}$ satisfy the KKT conditions  in (i)-(iv).
When $M \in \left\{\frac{N}{K}, \frac{2N}{K}, \ldots, \frac{(K-1)N}{K}\right\}$,  $\mathbf z^*$ given in  \eqref{eqn:uniform_caching}, $\theta^* =\frac{K+1}{\left(\frac{KM}{N}+1\right)^2}$,   $\nu^* =\frac{2K\frac{KM}{N}+K-\left(\frac{KM}{N}\right)^2}{\left(\frac{KM}{N}+1\right)^2}$ and
$\eta^*_{s}= {K \choose s}\frac{1}{s+1}\frac{K+1}{\left(\frac{KM}{N}+1\right)^2}\left(s-\frac{KM}{N}\right)^2$, $s \in \{0,1,\ldots,K\}$ satisfy the KKT conditions in (i)-(iv).
When $M=N$,  $\mathbf z^*$ given in  \eqref{eqn:uniform_caching},  any $\theta^* \in (0, \frac{1}{K+1}]$, $\nu^* =K\theta^* $ and    $\eta^*_{s}= {K \choose s}\left(\frac{K-s}{s+1}+\theta^* s -\nu^* \right)$, $s \in \{0,1,\ldots,K\}$  satisfy the KKT conditions in (i)-(iv).
Therefore,    \eqref{eqn:uniform_caching} is the unique  optimal solution  for  any    $M \in \left\{0, \frac{N}{K}, \frac{2N}{K}, \ldots, N\right\}$.
\end{itemize}
Substituting \eqref{eqn:uniform_caching} into \eqref{eqn:uniform_obj}, we can obtain  \eqref{eqn:Ali}.
Therefore, we complete the proof of Lemma~\ref{Lem:Ali}.



\begin{thebibliography}{10}
\providecommand{\url}[1]{#1}
\csname url@samestyle\endcsname
\providecommand{\newblock}{\relax}
\providecommand{\bibinfo}[2]{#2}
\providecommand{\BIBentrySTDinterwordspacing}{\spaceskip=0pt\relax}
\providecommand{\BIBentryALTinterwordstretchfactor}{4}
\providecommand{\BIBentryALTinterwordspacing}{\spaceskip=\fontdimen2\font plus
\BIBentryALTinterwordstretchfactor\fontdimen3\font minus
  \fontdimen4\font\relax}
\providecommand{\BIBforeignlanguage}[2]{{%
\expandafter\ifx\csname l@#1\endcsname\relax
\typeout{** WARNING: IEEEtran.bst: No hyphenation pattern has been}%
\typeout{** loaded for the language `#1'. Using the pattern for}%
\typeout{** the default language instead.}%
\else
\language=\csname l@#1\endcsname
\fi
#2}}
\providecommand{\BIBdecl}{\relax}
\BIBdecl

\bibitem{AliFundamental}
M.~A. Maddah-Ali and U.~Niesen, ``Fundamental limits of caching,'' \emph{IEEE
  Trans. Inf. Theory}, vol.~60, no.~5, pp. 2856--2867, May 2014.

\bibitem{AliDec}
------, ``Decentralized coded caching attains order-optimal memory-rate
  tradeoff,'' \emph{IEEE/ACM Trans. Netw.}, vol.~23, no.~4, pp. 1029--1040,
  Aug. 2015.

\bibitem{jin}
\BIBentryALTinterwordspacing
S.~Jin, Y.~Cui, H.~Liu, and G.~Caire, ``New order-optimal decentralized coded
  caching schemes with good performance in the finite file size regime,''
  \emph{CoRR}, vol. abs/1604.07648, 2016. [Online]. Available:
  \url{http://arxiv.org/abs/1604.07648}
\BIBentrySTDinterwordspacing

\bibitem{finite}
K.~Shanmugam, M.~Ji, A.~M. Tulino, J.~Llorca, and A.~G. Dimakis.,
  ``Finite-length analysis of caching-aided coded multicasting,'' \emph{IEEE
  Trans. Inf. Theory}, vol.~62, no.~10, pp. 5524--5537, Oct. 2016.

\bibitem{YuQian}
\BIBentryALTinterwordspacing
Q.~Yu, M.~A. Maddah{-}Ali, and A.~S. Avestimehr, ``The exact rate-memory
  tradeoff for caching with uncoded prefetching,'' \emph{CoRR}, vol.
  abs/1609.07817, 2016. [Online]. Available:
  \url{http://arxiv.org/abs/1609.07817}
\BIBentrySTDinterwordspacing

\bibitem{NonuniformDemands}
U.~Niesen and M.~A. Maddah-Ali, ``Coded caching with nonuniform demands,''
  \emph{IEEE Trans. Inf. Theory}, vol.~63, no.~2, pp. 1146--1158, Feb. 2017.

\bibitem{ji2015order}
M.~Ji, A.~M. Tulino, J.~Llorca, and G.~Caire, ``Order-optimal rate of caching
  and coded multicasting with random demands,'' \emph{IEEE Transactions on
  Information Theory}, vol.~63, no.~6, pp. 3923--3949, June 2017.

\bibitem{Jinbei}
J.~Zhang, X.~Lin, and X.~Wang, ``Coded caching under arbitrary popularity
  distributions,'' in \emph{Proc. IEEE ITA Workshop}, Feb. 2015, pp. 98--107.

\bibitem{Sinong}
S.~Wang, X.~Tian, and H.~Liu, ``Exploiting the unexploited of coded caching for
  wireless content distribution,'' in \emph{2015 International Conference on
  Computing, Networking and Communications (ICNC)}, Feb 2015, pp. 700--706.

\bibitem{KaiWan}
K.~Wan, D.~Tuninetti, and P.~Piantanida, ``On the optimality of uncoded cache
  placement,'' in \emph{Proc. IEEE ITW}, Sep. 2016, pp. 161--165.

\bibitem{bound17}
\BIBentryALTinterwordspacing
C.~Wang, S.~S. Bidokhti, and M.~A. Wigger, ``Improved converses and gap-results
  for coded caching,'' \emph{CoRR}, vol. abs/1702.04834, 2017. [Online].
  Available: \url{http://arxiv.org/abs/1702.04834}
\BIBentrySTDinterwordspacing

\bibitem{bound16}
\BIBentryALTinterwordspacing
C.~Wang, S.~H. Lim, and M.~Gastpar, ``A new converse bound for coded caching,''
  \emph{CoRR}, vol. abs/1601.05690, 2016. [Online]. Available:
  \url{http://arxiv.org/abs/1601.05690}
\BIBentrySTDinterwordspacing

\bibitem{Factor2}
\BIBentryALTinterwordspacing
Q.~Yu, M.~A. Maddah{-}Ali, and A.~S. Avestimehr, ``Characterizing the
  rate-memory tradeoff in cache networks within a factor of 2,'' \emph{CoRR},
  vol. abs/1702.04563, 2017. [Online]. Available:
  \url{http://arxiv.org/abs/1702.04563}
\BIBentrySTDinterwordspacing

\bibitem{HCD}
A.~Ramakrishnan, C.~Westphal, and A.~Markopoulou, ``An efficient delivery
  scheme for coded caching,'' in \emph{2015 27th International Teletraffic
  Congress}, Sep. 2015, pp. 46--54.

\bibitem{boyd2004convex}
S.~Boyd and L.~Vandenberghe, \emph{Convex optimization}.\hskip 1em plus 0.5em
  minus 0.4em\relax Cambridge university press, 2004.

\bibitem{LeeS13a}
\BIBentryALTinterwordspacing
Y.~T. Lee and A.~Sidford, ``Matching the universal barrier without paying the
  costs : Solving linear programs with {\~{o}}(sqrt(rank)) linear system
  solves,'' \emph{CoRR}, vol. abs/1312.6677, 2013. [Online]. Available:
  \url{http://arxiv.org/abs/1312.6677}
\BIBentrySTDinterwordspacing

\end{thebibliography}
\end{document}